\documentclass[english,journal]{IEEEtran}

\usepackage{xcolor}
\usepackage{bm}
\usepackage{stfloats}
\usepackage{mathtools}
\usepackage{enumitem}
\usepackage{siunitx}
\usepackage{algorithm}
\usepackage{algorithmic}
\usepackage[style=ieee, doi=false, isbn=false]{biblatex}
\usepackage{makecell}
\usepackage{siunitx}
\usepackage{graphicx}
\usepackage{url}      
\usepackage{amsmath}  
\usepackage{amsfonts,amssymb}
\usepackage{microtype}
\usepackage{babel}
\usepackage{csquotes}
\usepackage{comment}
\usepackage[normalem]{ulem}
\usepackage{enumitem}
\setlist{leftmargin=*}

\newtheorem{lemma}{{Lemma}}

\newtheorem{proposition}{{Proposition}}

\newtheorem{remark}{{Remark}}
\newtheorem{corollary}{{Corollary}}
\newenvironment{proof}{{\noindent\it Proof:}}{\hfill $\square$\par}

\definecolor{purple}{RGB}{128,0,128}

\newcommand{\rev}{\textcolor{black}}

\let\hat\widehat

\DeclareBibliographyCategory{bluepapers}

\AtEveryBibitem{%
  \ifcategory{bluepapers}{%
    \color{blue}%
  }{}
}

\begin{document}

\title{\fontsize{21pt}{25pt} Delay-Angle Information Spoofing for Channel State Information-Free Location-Privacy Enhancement}

\author{Jianxiu~Li,~\IEEEmembership{Member,~IEEE,} and~Urbashi~Mitra,~\IEEEmembership{Fellow,~IEEE}\vspace{-24pt}
\thanks{J. Li and U. Mitra are with the Department of Electrical and Computer Engineering, University of Southern California, Los Angeles, CA 90089, USA (e-mails: {jianxiul, ubli}@usc.edu).}
\thanks{{This paper was presented in part at 2024 IEEE International Conference on Communications (ICC) \cite{DAIS} and 2024 IEEE International Symposium on Information Theory (ISIT) \cite{DAIS_ISIT}}. This work has been funded in part by one or more of the following: NSF CCF-1817200, ARO W911NF1910269, ARO W911NF2410094, DOE DE-SC0021417, Swedish Research Council 2018-04359, NSF CCF-2008927, NSF CCF-2200221, NSF RINGS-2148313, NSF CCF-2311653, NSF A22-2666-S003, ONR 503400-78050, ONR N00014-15-1-2550, and ONR N00014-22-1-2363.}
}

\markboth{}
{LI AND MITRA: DELAY-ANGLE INFORMATION SPOOFING FOR CHANNEL STATE INFORMATION-FREE LOCATION-PRIVACY ENHANCEMENT} 

\maketitle

\begin{abstract}
In this paper, a delay-angle information spoofing (DAIS) strategy is proposed to enhance the location privacy at the physical layer. More precisely, the location-relevant delays and angles are artificially shifted without the aid of channel state information (CSI) at the transmitter, such that the location perceived by the eavesdropper is incorrect and distinct from the true one. By leveraging the intrinsic structure of the wireless channel, a precoder is designed to achieve DAIS while the legitimate localizer can remove the obfuscation via securely {receiving} a modest amount of information, \textit{i.e.,} the delay-angle shifts. A lower bound on eavesdropper's localization error is derived, {revealing} that location privacy is enhanced  not only due to  estimation error, but also by the \textit{geometric mismatch} introduced by DAIS. Furthermore, the lower bound is explicitly expressed as a function of the delay-angle shifts, characterizing performance trends and providing the appropriate design of these shift parameters. The statistical hardness of maliciously inferring the delay-angle shifts by a single-antenna eavesdropper as well as the challenges for a multi-antenna eavesdropper are investigated to assess the robustness of the proposed DAIS strategy. Numerical results show that the proposed DAIS strategy results in more than $15$ dB performance degradation for the eavesdropper as compared with that for the legitimate localizer at high signal-to-noise ratios, and provides more effective location-privacy enhancement than the prior art.
\end{abstract}
\vspace{-3pt} 

\IEEEpeerreviewmaketitle

\vspace{-10pt}
\section{Introduction}
\label{sec:intro}
\IEEEPARstart{P}{roviding} accurate location information is a crucial service for fifth generation of wireless technologies and beyond \cite{5Glocalizationtutorial}, and supports {many} applications, such as smart factories \cite{smarfactory5G}, autonomous driving \cite{V2X5G} and assisted living \cite{assistedliving}. {For example, thanks to the large bandwidth and limited multipath,  millimeter wave (mmWave) signals \cite{Rappaport} have been widely employed to achieve centimeter-level localization accuracy with only a single legitimate localizer in a multi-antenna system \cite{Shahmansoori,Zhou,FascistaMISOML,LOCMAN},} where the location-relevant parameters, \textit{e.g.}, delay-angle information associated with each wireless path, can be precisely estimated to infer the location of the user equipment (UE). However, the main focus of the existing localization algorithm designs, such as \cite{Shahmansoori,Zhou,FascistaMISOML,LOCMAN}, is how to leverage the features of the wireless signals to improve estimation accuracy, without considering location-privacy preservation. Location privacy of the UE can be threatened due to the broadcast nature of electromagnetic wave propagation. {If the eavesdropper has the access to these wireless signals, the UE’s location will be exposed, which can be used to potentially infer more private information, \textit{e.g.,} personal preferences \cite{Ayyalasomayajula,li2023fpi}.}

To reduce privacy leakage at the physical layer, the actual channel state, or the statistics of the channel, is usually exploited for privacy enhancement designs \cite{Goel,Tomasin2,dacosta2023securecomm,Goztepesurveyprivacy,Checa, Ayyalasomayajula,li2023fpi,Ardagnaobfuscation}. One approach for preserving location privacy is to protect the location-relevant information from being easily snooped from the wireless signals. For example, prior strategies \cite{Tomasin2,Checa} either inject artificial noise that is in the \textit{null space} of the legitimate channel to decrease the received signal-to-noise ratio (SNR) for the eavesdropper \cite{Tomasin2} or hide key delay-angle information via transmit beamforming \cite{Checa}. Another method is to perturb location-relevant information as location inference usually relies on correct knowledge of such information {\cite{Ayyalasomayajula,DOLOS}}. In \cite{Ayyalasomayajula}, a beamformer is designed to virtually increase the delay of the line-of-sight (LOS) path such that one of the non-line-of-sight (NLOS) paths appears to have the smallest delay and thus is incorrectly treated as the LOS for inferring UE’s location. An improved beamforming design is recently provided in \cite{DOLOS}, with reduced amount of channel state information (CSI) needed to achieve the same goal as in \cite{Ayyalasomayajula}. However, all of these designs \cite{Tomasin2,Checa,Ayyalasomayajula,DOLOS} rely on accurate CSI to be present at the UE; acquiring such knowledge strongly increases the overhead for the resource-limited UE.

To prevent the location-relevant information from being accurately estimated without the aid of CSI, {\textit{fake paths} are designed in \cite{li2023fpi}} and virtually injected to the channel via a precoder design. Motivated by the
studies of stability of the super-resolution problem \cite{Ankur,LiTIT,candes2014towards,LIACHA,ANDE,LOCMAN}, {the Fisher information matrix (FIM) associated with the eavesdropper's inference is analyzed to show that accurately estimating the location-relevant parameters becomes statistically harder in the presence of fake path injection (FPI) \cite{li2023fpi}}, especially when the injected fake paths are highly correlated with the true paths. The CSI-free FPI strategy is further investigated for a single-input-multiple-output system and a scenario with multiple single-antenna eavesdroppers in \cite{FPISIMO} and \cite{SISOMultipleUD}, respectively. A FPI-based deceptive jamming design is also examined in \cite{yildirim2024deceptive}. Although the effectiveness of FPI has been validated in \cite{li2023fpi,FPISIMO,SISOMultipleUD, yildirim2024deceptive}, such a design does not directly spoof the location-relevant information itself, so location snooping is still possible at high SNRs when the bandwidth and the number of antennas are quite large. Additionally, if the structure of the designed precoder is unfortunately leaked to the eavesdropper, the associated precoding matrix can be inferred with a sufficiently large number of measurements, undermining the efficacy of the FPI strategy \cite{li2023fpi}. 

In this paper, motivated by the obfuscation techniques in \cite{Ardagnaobfuscation,Ayyalasomayajula}, we propose a  \textit{delay-angle information spoofing} (DAIS) strategy to enhance location privacy. Distinct from \cite{Ayyalasomayajula}, CSI is not needed herein. By artificially shifting the delays and angles with the proposed DAIS strategy, the UE’s location as perceived by the eavesdropper is \textit{virtually} moved to an {\bf incorrect} location that can be far from the true one. In contrast, such \textit{geometric mismatch} is not introduced with the FPI strategy in \cite{li2023fpi} though a challenging estimation framework is created therein. As a result of the differences, in this work, we still analyze the associated Fisher information, but develop a \textit{misspecified} Cram\'{e}r-Rao bound (MCRB) \cite{Fortunatimismatchsurvey} to characterize eavesdropper’s estimation error with DAIS, versus a true Cram\'{e}r-Rao bound (CRB) in \cite{li2023fpi}. Our theoretical analysis characterizes the amount of obfuscation possible via the proposed DAIS strategy. \rev{We note that the proposed strategy does not rely on the channel asymmetry and can significantly degrade the eavesdropper's localization accuracy regardless of whether the eavesdropper is more powerful or less powerful.} The main contributions of this paper are: 
\begin{enumerate}
\item  The DAIS framework is proposed, where all the location-relevant delays and angles are shifted.
\item To {mislead} the eavesdropper via the shifted delays and angles, a {new} CSI-free precoding strategy is proposed, {where the delay-angle \textit{shifts} \rev{(which are like secret keys~\cite{Schaefer,ZhangQ})} are securely shared with the legitimate localizer to maintain legitimate localizer’s localization accuracy}.
\item A MCRB is derived in the presence of DAIS, with a closed-form expression for the pseudo-true (incorrect) locations, theoretically validating the efficacy of the proposed scheme.
\item Under an assumption of the orthogonality among the paths, the derived MCRB is explicitly expressed as a function of the delay-angle shifts in a closed-form, suggesting the sensitivity to these key parameters as well as design rules for their selection.
\item The possible leakage of the structure of the designed precoder to a single-antenna eavesdropper and the challenges for a multi-antenna eavesdropper are investigated, which shows the increased robustness of the proposed scheme in contrast to \cite{li2023fpi}.
\item Numerical comparisons for different methods of location privacy enhancement are provided, showing that more than $15$ dB performance degradation is achievable at high SNRs for the eavesdropper, due to the proposed DAIS method. In terms of the eavesdropper's localization accuracy, the root-mean-square error (RMSE) with  DAIS can be more than $150$ m, which outperforms the prior art \cite{li2023fpi,Ayyalasomayajula}.
\end{enumerate}

The present work completes our previous works \cite{DAIS,DAIS_ISIT} with  explicit proofs and new comparisons with prior art \cite{li2023fpi,Ayyalasomayajula}. Distinct from \cite{DAIS,DAIS_ISIT} where the focus is on the performance degradation for a single-antenna eavesdropper, a multi-antenna eavesdropper is also analyzed herein. We note that, though \cite{Ayyalasomayajula} proposes to add an extra delay to the LOS path to degrade the localization accuracy for a multi-antenna eavesdropper, {angle information is not spoofed in \cite{Ayyalasomayajula} and the eavesdropper still can leverage angle information to infer UE’s location once the LOS path can be identified.} In contrast, the challenges resulting from the use of DAIS in this environment are investigated in this paper to show the robustness of this new design. 

The rest of this paper is organized as follows. Section~\ref{sec:signalmodel} introduces the system model used for localization. Section~\ref{sec:csifreedesign} presents the proposed DAIS strategy for location-privacy enhancement, where the delays and angles are shifted via a CSI-free precoder design. In Section~\ref{sec:analysis}, lower bounds on the localization error are derived and analyzed for both the legitimate localizer and the eavesdropper, proving the enhancement of location privacy with DAIS. In addition, the impact of the delay-angle shifts in the DAIS design is also investigated in this section, by further simplifying the lower bounds derived for eavesdropper’s localization accuracy. Section \ref{sec:DAISworsecase} characterizes the robustness of the proposed DAIS strategy, where the statistical hardness of avoiding DAIS for a single-antenna eavesdropper as well as the challenges for a multi-antenna eavesdropper are specifically studied. Numerical results are provided in Section~\ref{sec:sim} to validate the theoretical analyses and highlight the performance degradation caused by DAIS for the eavesdropper relative to other schemes. Conclusions are drawn in Section~\ref{sec:con}. Appendices \ref{sec:proofpseudotruepos}, \ref{sec:proofdegradation}, \ref{sec:proofEveLBsimplifed}, \ref{sec:prooflowerboundpos}, and \ref{sec:proofstrucleakge} provide the proofs for the key technical results.

We use the following notation. Scalars are denoted by lower-case letters $x$ and column vectors by bold letters $\bm{x}$. The $i$-th element of $\bm{x}$ is denoted by $\bm{x}[i]$. Matrices are denoted by bold capital letters $\bm X$ and $\boldsymbol{X}[i, j]$ is the ($i$, $j$)-th element of $\bm{X}$. $\boldsymbol{I}_{l}$ is reserved for a $l\times l$ identity matrix. The operators $\|\bm  x\|_{2}$, $|x|$, $\mathfrak{R}\{x\}$, $\mathfrak{I}\{x\}$, $\lfloor{x}\rfloor$, and $\operatorname{diag}(\mathcal{A})$ stands for the $\ell_2$ norm of $\bm x$, the magnitude of $x$, the real part of $x$, the imaginary part of $x$, the largest integer that is less than $x$, and a diagonal matrix whose diagonal elements are given by $\mathcal{A}$, respectively. The operator $(x)_{(t_1,t_2]}$ with $t_1<t_2$ is defined as $(x)_{(t_1,t_2]}\triangleq x-\left\lfloor\frac{x-t_1}{t_2-t_1}\right\rfloor(t_2-t_1)$
and $\mathbb{E}\{\cdot\}$ is used for the expectation of a random variable. The operators $(\cdot)^\mathrm{T}$, $(\cdot)^{\mathrm{H}}$ and $(\cdot)^{-1}$, are defined as the transpose, the conjugate transpose, and the inverse of a vector or matrix, respectively.

\vspace{-10pt} 
\section{System Model}\label{sec:signalmodel}%\vspace{-3pt} 

As shown in Figure \ref{fig:sys}, we consider a system model similar to \cite{li2023fpi}, where a legitimate localizer (Bob), at a location $\bm q^{\star}=[q^{\star}_x, q^{\star}_y]^{\mathrm{T}} \in \mathbb{R}^{2\times 1}$, serves a UE (Alice) at an unknown position $\bm p^{\star} = [p^{\star}_x, p^{\star}_y]^{\mathrm{T}} \in \mathbb{R}^{2\times 1}$. To provide the location-based services, after Alice transmits pilot signals through a public channel, Bob can infer Alice's location from the received signal. Unfortunately, there is an eavesdropper (Eve), at a position $\bm z^{\star}=[z^{\star}_x, z^{\star}_y]^{\mathrm{T}}\in\mathbb{R}^{2\times 1}$, who can eavesdrop on the public channel to also estimate Alice's location. We assume both Bob and Eve know the pilot signals as well as their own locations so Eve’s malicious inference jeopardizes Alice's location-privacy if no location-privacy preservation mechanisms are adopted. 

Herein, mmWave multiple-input-single-output (MISO) orthogonal frequency-division multiplexing (OFDM) signaling is used for the transmissions. Accordingly, Alice has $N_t$ antennas while both Bob and Eve are equipped with a single antenna\footnote{\rev{The goal of this work is {\bf not to design} a localization algorithm, but rather a  \emph{location-privacy enhancement strategy}. The proposed design does not assume that Bob must use a single antenna. Using certain shared information that is inaccessible to Eve, Bob still can accurately infer Alice's location, even when he has multiple antennas, which will be elaborated in Sections \ref{sec:csifreedesign} and \ref{sec:analysis}.} The analysis for a multi-antenna \rev{Eve} will be provided in \rev{Sections} \ref{sec:DAISworsecase} and \rev{\ref{sec:sim}}.}. Denoting by $N$ and $G$ the number of sub-carriers and the number of the transmitted signals, respectively, we express the $g$-th symbol transmitted over the $n$-th sub-carrier as $x^{(g,n)}$ and the corresponding beamforming vector as $\bm f^{(g,n)}\in\mathbb{C}^{N_t\times1}$. Then, the pilot signal can be written as $\boldsymbol{s}^{(g,n)}\triangleq \boldsymbol{ f}^{(g,n)}x^{(g,n)}\in\mathbb{C}^{N_t\times 1}
$. Assume that each pilot signal transmitted over the $n$-th sub-carrier are independent and identically distributed and we have $\mathbb{E}\{\boldsymbol{s}^{(g,n)}(\boldsymbol{s}^{(g,n)})^{\mathrm{H}}\}= \frac{1}{N_t}\bm I_{N_t}$. The received signal is given by \vspace{-4pt} 
\begin{equation}
{y}^{(g,n)}=\boldsymbol{h}^{(n)} \boldsymbol{s}^{(g,n)}+{w}^{(g,n)},\label{eq:rsignal}\vspace{-4pt} 
\end{equation}
for $n = 0, 1, \cdots, N-1$ and $g = 1, 2, \cdots, G$, where $\bm h^{(n)}\in \mathbb{C}^{1\times N_t}$ is the $n$-th sub-carrier public channel vector while ${w}^{(g,n)}\sim \mathcal{CN}({0},\sigma^2)$ represents independent, zero-mean, complex Gaussian noise with variance $\sigma^2$.

\begin{figure}[t]
\centering
\includegraphics[scale=0.42]{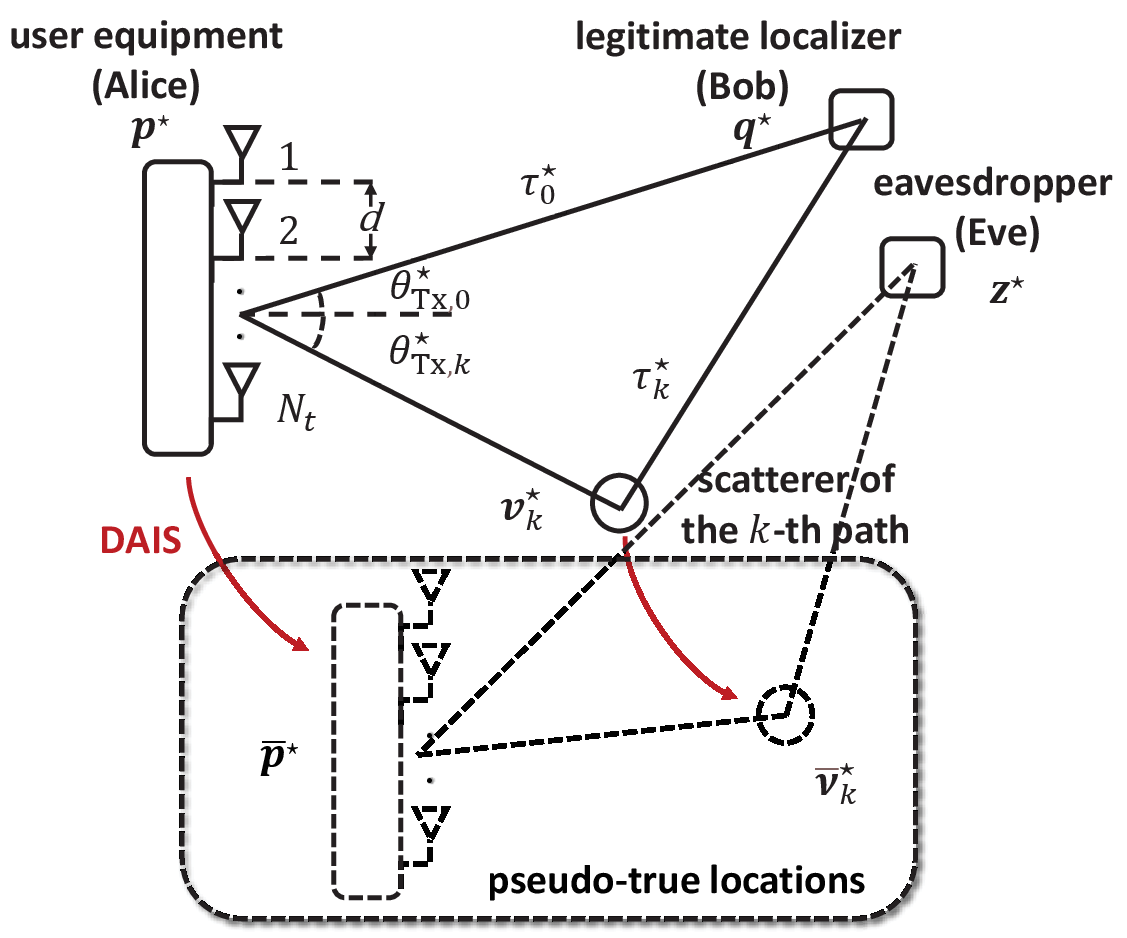}\vspace{-6pt}
\caption{System model.}\vspace{-15pt}
\label{fig:sys}
\end{figure}

We assume that there exist $K$ NLOS paths in the channel, in addition to an available LOS path. The $k$-th NLOS path is produced by a scatterer at an \textit{unknown} position $\bm v^{\star}_k = [v^{\star}_{k,x},v^{\star}_{k,y}]^{\mathrm{T}}\in\mathbb{R}^{2\times 1}$, with $k = 1,2,\cdots,K$. Denote by $c$, $\varphi_c$, $B$, and $T_s\triangleq\frac{1}{B}$, the speed of light, carrier frequency, bandwidth, and sampling period, respectively. A narrowband channel is considered in this paper, \textit{i.e.}, $B\ll \varphi_c$, and the public channel vector $\bm h^{(n)}$ can be modeled as \vspace{-5pt} \cite{li2023fpi,FascistaMISOMLCRB,FascistaMISOML}
\begin{equation}
\boldsymbol{h}^{(n)}\triangleq\sum_{k=0}^{K}\gamma^{\star}_k e^{\frac{-j 2\pi n\tau^{\star}_k}{N T_{s}}}\boldsymbol{ \alpha}\left(\theta^{\star}_{\mathrm{Tx},k}\right)^{\mathrm{H}},
\label{eq:channel_subcarrier}\vspace{-5pt} 
\end{equation}
where $k=0$ corresponds to the LOS path while $\gamma^{\star}_k$, $\tau^{\star}_k$, and $\theta^{\star}_{\mathrm{Tx},k}$ represent the complex channel coefficient, the time-of-arrival (TOA), and the angle-of-departure (AOD) of the $k$-th path, respectively. The steering vector $\boldsymbol{ \alpha}\left(\theta^{\star}_{\mathrm{Tx},k}\right)\in\mathbb{C}^{N_t\times1}$ is defined as $\bm\alpha(\theta^{\star}_{\mathrm{Tx},k})\triangleq\left[1, e^{-j\frac{2\pi d\sin(\theta^{\star}_{\mathrm{Tx},k})}{\lambda_c}}, \cdots, e^{-j\frac{2\pi(N_t-1)d\sin(\theta^{\star}_{\mathrm{Tx},k})}{\lambda_c}}\right]^{\mathrm{T}}$, for $k=0,1,\cdots,K$, where $\lambda_c\triangleq\frac{c}{\varphi_c}$ is the wavelength and $d$ is the distance between antennas, designed as $d=\frac{\lambda_c}{2}$. Define $\bm v^{\star}_0\triangleq \bm q^{\star}$ (or $\bm v^{\star}_0\triangleq \bm z^{\star}$) for Bob (or Eve). From the geometry, the TOA and AOD of the $k$-th path are given by\vspace{-5pt} 
\begin{equation}\label{eq:geometry}
\begin{aligned}
\tau^{\star}_{k} &=\frac{\left\|\boldsymbol{v}^{\star}_0-\boldsymbol{v}^{\star}_{k}\right\|_{2} +\left\|\boldsymbol{p}^{\star}-\boldsymbol{v}^{\star}_{k}\right\|_{2}} {c}\\
\theta^{\star}_{\mathrm{Tx}, k} &=\arctan \left(\frac{v^{\star}_{k,y}-p^{\star}_{y}} {v^{\star}_{k,x}-p^{\star}_{x}}\right),\vspace{-3pt} 
\end{aligned}
\end{equation} 
where we assume $0<\tau^{\star}_k\leq T_{\mathrm{cp}}\leq NT_s$ and $-\frac{\pi}{2}<\theta^{\star}_{\mathrm{Tx},k}\leq\frac{\pi}{2}$ with $T_{\mathrm{cp}}$ being the duration of the cyclic prefix (CP). In addition, it is assumed that the orientation angle of Alice's antenna array and the clock bias are known to both Bob and Eve as prior information; without loss of generality, these parameters are set to zero in this paper\footnote{\rev{These assumptions simplify the signal model without loss of generality, and enable better illustration of the location privacy enhancement;  our proposed strategy does not rely on these assumptions.}}.

{Given the noise level characterized by $\sigma^2$}, once the received signals are collected, Alice's location can be inferred with the pilot signals. \rev{On} the one hand, to enhance location privacy, we aim to increase Eve’s localization error. On the other hand, Alice needs to transmit the \textit{shared information} over a secure channel to ensure Bob's localization accuracy. Note that, CSI is assumed to be unavailable to Alice.

\vspace{-6pt} 
\section{Delay-Angle Information Spoofing for Location-Privacy Enhancement}\label{sec:csifreedesign} 
For the model-based localization designs, such as \cite{Shahmansoori,LOCMAN,FascistaMISOMLCRB,FascistaMISOML}, channel parameters are typically inferred in the first stage and then locations are estimated from the location-relevant channel parameters, \textit{i.e.,} TOAs and AODs, according to the geometry given in Equation \eqref{eq:geometry}.  More precisely, without any location-privacy preservation mechanisms, the locations can be derived as, \vspace{-4pt} 
\begin{equation}\label{eq:locationestimate}
    \begin{aligned} 
        {\boldsymbol{p}}^{\star} &= \boldsymbol{z}^\star - c\tau^\star_{0}[\cos({\theta}^\star_{\mathrm{Tx},0}),\sin(\theta^\star_{\mathrm{Tx},0})]^{\mathrm{T}},\\ 
        {\boldsymbol{v}}^\star_{k} &=\frac{1}{2}{b}^{\star}_{k}[\cos({\theta}^\star_{\mathrm{Tx},k}),\sin({\theta}^\star_{\mathrm{Tx},k})]^{\mathrm{T}}+{\boldsymbol{p}}^{\star},\vspace{-4pt} 
    \end{aligned} 
\end{equation}
where \vspace{-4pt} 
\begin{equation}
    \begin{aligned}          {b}^{\star}_{k}=\frac{\left(c\tau^\star_k\right)^2-\left(z^\star_x-{p}^{\star}_{x}\right)^2-\left(z^\star_y-{p}^{\star}_{y}\right)^2}{c\tau^\star_k-\left(z^\star_x-{p}^{\star}_{x}\right)\cos({\theta}^\star_{\mathrm{Tx},k})-\left(z^\star_y-{p}^{\star}_{y}\right)\sin({\theta}^\star_{\mathrm{Tx},k})}, \vspace{-4pt} 
    \end{aligned}%\vspace{-2pt}
\end{equation}
with $k=1,2,\cdots,K$. 

As observed in Equation \eqref{eq:locationestimate}, high localization accuracy not only relies on super-resolution channel estimation, but also requires the knowledge of the geometric model. Hence, distinct from making Eve's channel estimation statistically harder with a FPI design \cite{li2023fpi}, in this section, we first propose a new CSI-free framework for location-privacy enhancement, where the delays and angles of the paths in Eve's channel are artificially shifted to {mislead} Eve to perceive a mismatched geometric model. Then, a precoding design is provided as one approach to efficiently achieve DAIS. 

\vspace{-6pt} 
\subsection{Delay-Angle Information Spoofing}\label{subsec:dias}
Let $\Delta_\tau$ and $\Delta_\theta$ represent two constants used for the DAIS design. To prevent Alice's location from being accurately inferred by Eve from the estimate of the location-relevant channel parameters, the TOAs and AODs of the paths in Eve’s channel are shifted according to $\Delta_\tau$ and $\Delta_\theta$, respectively, as  \vspace{-4pt} 
\begin{equation}\label{eq:mismatchgeometry}
\begin{aligned}
\bar{\tau}_{k}& = \left(\tau^{\star}_{k} + \Delta_\tau\right)_{(0,NT_s]}\\
\bar{\theta}_{\mathrm{Tx}, k}&= \arcsin\left({\left(\sin(\theta^{\star}_{\mathrm{Tx}, k})+\sin(\Delta_\theta)\right)_{\left(-1,1\right]}}\right).\vspace{-4pt} 
\end{aligned}
\end{equation} 
By shifting the TOAs and AODs during the transmission of the pilot signals, the proposed DAIS scheme misleads Eve into treating an incorrect physical location as the true one if she exploits the incorrect geometric model in Equation \eqref{eq:geometry} for localization. The degraded localization accuracy will be analyzed in Section \ref{sec:degradationdais}. We note that, since shifting delay and angles does not rely on the channel parameters, location-privacy preservation with the proposed DAIS strategy is CSI-free, though the optimal design of such shifts requires the CSI, which will be shown in Section \ref{subsec:parameterdesign}. {However, even without CSI, we can propose effective shift designs.} 

\vspace{-9pt}
\subsection{Precoder Design}
Alice employs the mmWave MISO OFDM signaling; to achieve DAIS, we design a precoding matrix $\bm\Phi^{(n)}\in\mathbb{C}^{N_t\times N_t}$ for Alice as\footnote{{This new precoder has some similarities to the precoding matrix in \cite{li2023fpi}, defined as ${\bm \Phi}_{\text{FPI}}^{(n)} \triangleq\bm I_{N_t} + e^{-\frac{2\pi n \bar\delta_\tau}{NT_s}} \operatorname{diag}\left(\bm \alpha\left(\bar\delta_{\theta_{\text{TX}}}\right)^{\mathrm{H}}\right)$, where $\bar\delta_\tau$ and $\bar\delta_{\theta_{\text{TX}}}$ are two design parameters for FPI. However, given the subtle differences, the design of the parameters, \textit{i.e.,} $\Delta_\tau$, $\Delta_{\theta}$ versus $\bar\delta_\tau$, $\bar\delta_{\theta_{\text{TX}}}$, and the associated analyses in these two works are quite distinct.}}\vspace{-4pt}
\begin{equation}\label{eq:daisbeamformer}
    {\bm \Phi}^{(n)} \triangleq  e^{-\frac{{j}2\pi n \Delta_\tau}{NT_s}} \operatorname{diag}\left(\bm \alpha\left(\Delta_\theta\right)^{\mathrm{H}}\right),\vspace{-5pt}
\end{equation}
with $n=0,1,\cdots,N-1$. Then, the signal received through the public channel can be re-expressed as \vspace{-5pt}
\begin{equation}
\begin{aligned}
\bar{y}^{(g,n)}
&=\boldsymbol{h}^{(n)}{{\bm \Phi}^{(n)}}{\boldsymbol{ s}}^{(g,n)}+{w}^{(g,n)},\\
&=\boldsymbol{h}^{(n)}e^{-\frac{{j}2\pi n \Delta_\tau}{NT_s}} \operatorname{diag}\left(\bm \alpha\left(\Delta_\theta\right)^{\mathrm{H}}\right){\bm s}^{(g,n)}+{w}^{(g,n)}\\
&=\sum_{k=0}^{K}\gamma^{\star}_k e^{-\frac{j 2\pi n\bar{\tau}_k}{N T_{s}}}\boldsymbol{ \alpha}\left(\bar{\theta}_{\mathrm{Tx},k}\right)^{\mathrm{H}}{\bm s}^{(g,n)}+{w}^{(g,n)}\\
&=\bar{\boldsymbol{h}}^{(n)}{\bm s}^{(g,n)}+{w}^{(g,n)},\vspace{-5pt}
\end{aligned}\label{eq:daisrsignal}
\end{equation}
where \vspace{-4pt}
\begin{equation}
\bar{\boldsymbol{h}}^{(n)}\triangleq\sum_{k=0}^{K}\gamma^{\star}_k e^{\frac{-j 2\pi n\bar{\tau}_k}{N T_{s}}}\boldsymbol{ \alpha}\left(\bar{\theta}_{\mathrm{Tx},k}\right)^{\mathrm{H}} \vspace{-5pt}
\end{equation} 
represents a \textit{virtual channel} for the $n$-th sub-carrier, constructed based on the original channel ${\boldsymbol{h}}^{(n)}$ with shifted TOAs and AODs.

\vspace{-9pt}
\subsection{Eve’s Localization}%\vspace{-1pt}
With the knowledge of $\bar{y}^{(g,n)}$ and ${\bm s}^{(g,n)}$, Eve will incorrectly believe that $\bar{\boldsymbol{h}}^{(n)}$ is the channel to be estimated\footnote{Using the proposed precoding strategy in Equation \eqref{eq:daisbeamformer} for DAIS, the physical public channel ${\boldsymbol{h}}^{(n)}$ is not fundamentally changed. Therefore, we do not need to adjust the duration of the CP to avoid inter-symbol interference and maintain the orthogonality of the sub-carriers. The signal model in Equation \eqref{eq:rsignal} still holds.}; after channel estimation, Eve cannot accurately infer Alice's true location using the mismatched geometric model in Equation \eqref{eq:geometry}. It will be shown in Section \ref{subsec:eveerror} that shifting the TOAs and AODs virtually moves Alice and the $k$-th scatterer to other positions, denoted as $\bar{\bm p}^\star\triangleq[\bar{p}^\star_x,\bar{p}^\star_y]^\mathrm{T}\in\mathbb{R}^{2\times1}$ and $\bar{\bm v}^\star_k\triangleq[\bar{v}^\star_{k,x},\bar{v}^\star_{k,y}]^\mathrm{T}\in\mathbb{R}^{2\times1}$, respectively, with $k=1,2,\cdots,K$. 

\vspace{-9pt}
\subsection{Bob’s Localization}%\vspace{-1pt}
In contrast, since we assume that Bob receives the shared information $\bm\Delta\triangleq[\Delta_\tau,\Delta_\theta]^\mathrm{T}\in\mathbb{R}^{2\times1}$ through a secure channel that is inaccessible by Eve\footnote{Error in the shared information can reduce Bob's estimation accuracy as well, but the study of such error is beyond the scope of this paper.}, Bob can construct effective pilot signals $\bar{\bm s}^{(g,n)}\triangleq {\bm\Phi^{(n)}}{\bm s}^{(g,n)}\in\mathbb{C}^{N_t\times1}$ for estimation. By leveraging the original signal model in Equation \eqref{eq:rsignal} with the knowledge of $\bar{\bm s}^{(g,n)}$ for localization, Bob is not measurably affected by the proposed DAIS scheme and he can maintain his localization accuracy. We note that not all pilot signals yield the same estimation accuracy, thus Bob's performance could either be slightly degraded or improved in DAIS.

\vspace{-6pt} 
\section{Localization Accuracy with Delay-Angle Information Spoofing}\label{sec:analysis}
To show the enhanced location privacy with the proposed DAIS strategy, Fisher information based lower bounds for Bob's and Eve’s localization error are derived and analyzed in this section. The analysis is not restricted to specific estimators. To be more concrete, the \textit{effective Fisher information} \cite{Tichavskyefim} is introduced to characterize channel estimation accuracy. Then, the CRB \cite{Scharf} for Bob's localization error is provided accordingly, which is used for a theoretical comparison between strategies. Due to the geometric mismatch introduced by DAIS, the MCRB \cite{Fortunatimismatchsurvey} is leveraged for the analysis of Eve's localization error, which is further investigated to show the efficacy of the proposed strategy.

\vspace{-9pt} 
\subsection{Effective Fisher Information for Channel Estimation}
Define $\bar{\bm\xi}\triangleq\left[\bar{\bm\tau}^\mathrm{T},\bar{\bm\theta}_{\mathrm{Tx}}^{\mathrm{T}},\mathfrak{R}\{\bm\gamma^{\star}\},\mathfrak{I}\{\bm\gamma^{\star}\}\right]^{\mathrm{T}}\in\mathbb{R}^{4(K+1)\times 1}$ as a vector of the unknown channel parameters, where $\bar{\bm\tau}\triangleq\left[\bar{\tau}_0, \bar{\tau}_1\cdots, \bar{\tau}_K\right]^{\mathrm{T}}\in\mathbb{R}^{(K+1)\times 1}$, $\bar{\bm\theta}_{\mathrm{Tx}}\triangleq\left[\bar{\theta}_{\mathrm{Tx},0}, \bar{\theta}_{\mathrm{Tx},1}\cdots, \bar{\theta}_{\mathrm{Tx},K}\right]^{\mathrm{T}}\in\mathbb{R}^{(K+1)\times 1}$, and $\bm\gamma^{\star} \triangleq [\gamma^{\star}_0,\gamma^{\star}_1,\cdots,\gamma^{\star}_K]^{\mathrm{T}}\in\mathbb{R}^{(K+1)\times 1}$. The \rev{FIM} for the estimation of $\bar{\bm \xi}$, denoted as $\bm J_{\bar{\bm \xi}}\in\mathbb{R}^{4(K+1)\times 4(K+1)}$, is given by \cite{Scharf} \vspace{-4pt} 
\begin{equation}\label{eq:FIMce}
    \bm J_{\bar{\bm \xi}} = \frac{2}{\sigma^2}\sum_{n=0}^{N-1}\sum_{g=1}^{G}\mathfrak{R}\left\{\left(\frac{\partial\overline{\varpi}^{(g,n)}}{\partial \bar{\bm\xi}}\right)^{*}\frac{\partial \overline{ \varpi}^{(g,n)}}{\partial \bar{\bm\xi}}\right\}, \vspace{-4pt} 
\end{equation}
where $\overline{\varpi}^{(g,n)} \triangleq \bar{\boldsymbol{h}}^{(n)}\bm s^{(g,n)}$.

Let $\bar{\bm\eta}\triangleq[\bar{\bm\tau}^\mathrm{T},\bar{\bm\theta}_{\mathrm{Tx}}^{\mathrm{T}}]^{\mathrm{T}}\in\mathbb{R}^{2(K+1)\times1}
$ represent the location-relevant channel parameters in the presence of DAIS. We partition the FIM $\bm J_{\bar{\bm \xi}}$ as $\bm J_{\bar{\bm \xi}}= \begin{bmatrix}\bm J_{\bar{\bm \xi}}^{(1)} &\bm J_{\bar{\bm \xi}}^{(2)}\\\bm J_{\bar{\bm \xi}}^{(3)}&\bm J_{\bar{\bm \xi}}^{(4)}\end{bmatrix}$, with $\bm J_{\bar{\bm \xi}}^{(m)}\in\mathbb{R}^{2(K+1)\times2(K+1)}$, for $m=1,2,3,4$. To analyze the localization accuracy, the channel coefficients are considered as nuisance parameters and accordingly the effective FIM for the estimation of the location-relevant channel parameters $\bar{\bm\eta}$ can be derived as \cite{Tichavskyefim}\vspace{-4pt} 
\begin{equation} \label{eq:efim}
    \bm J_{\bar{\bm \eta}} = \bm J_{\bar{\bm \xi}}^{(1)}-\bm J_{\bar{\bm \xi}}^{(2)}\left(\bm J_{\bar{\bm \xi}}^{(4)}\right)^{-1}\bm J_{\bar{\bm \xi}}^{(3)}\in\mathbb{R}^{2(K+1)\times 2(K+1)}. \vspace{-4pt} 
\end{equation}
Using the proposed DAIS method for location-privacy enhancement, the localization accuracy for Bob and Eve will be studied in the following subsections, respectively.

\vspace{-9pt} 
\subsection{Bob's Localization Error} 
Since Bob has the access to the shared information $\bm\Delta$, there is no error introduced into Bob's knowledge of the geometric model; he still can leverage the geometric model in Equation \eqref{eq:geometry} to infer Alice's true location after estimating the true channel based on the received signals $\bar{y}^{(g,n)}$ and effective pilot signals $\bar{\boldsymbol{s}}^{(g,n)}$. Similar to Equation \eqref{eq:efim}, the effective FIM for the estimation of true location-relevant channel parameters $\bm\eta^\star\triangleq[(\bm\tau^\star)^{\mathrm{T}},(\bm\theta^\star_{\mathrm{Tx}})^{\mathrm{T}}]\in\mathbb{R}^{2(K+1)\times 1}$ with ${\bm\tau}^\star\triangleq\left[{\tau}^\star_0, {\tau}^\star_1\cdots, {\tau}^\star_K\right]^{\mathrm{T}}\in\mathbb{R}^{(K+1)\times 1}$ and ${\bm\theta}^\star_{\mathrm{Tx}}\triangleq\left[{\theta}^\star_{\mathrm{Tx},0}, {\theta}^\star_{\mathrm{Tx},1}\cdots, {\theta}^\star_{\mathrm{Tx},K}\right]^{\mathrm{T}}\in\mathbb{R}^{(K+1)\times 1}$ can be derived, which is denoted as $\bm J_{{\bm \eta}^\star}\in\mathbb{R}^{2(K+1)\times 2(K+1)}$. 

Then, the FIM for Bob's localization, denoted as $\bm J_{\bm \phi^{\star}}\in\mathbb{R}^{2(K+1)\times 2(K+1)}$, is given by \vspace{-3pt}
\begin{equation}\label{eq:FIMloc}
    \bm J_{\bm \phi^{\star}} = \bm\Pi_{\bm \phi^{\star}}^\mathrm{T}\bm J_{{\bm \eta}^\star}\bm\Pi_{\bm \phi^{\star}}, \vspace{-3pt}
\end{equation}
where $\bm\phi^{\star}\triangleq [(\boldsymbol{p}^{\star})^{\mathrm{T}},(\boldsymbol{v}^{\star}_1)^{\mathrm{T}},(\boldsymbol{v}^{\star}_2)^{\mathrm{T}},\cdots,(\boldsymbol{v}^{\star}_K)^{\mathrm{T}}]^{\mathrm{T}}\in\mathbb{R}^{2(K+1)\times1}$ is a vector of the true locations of Alice and scatterers while $\bm\Pi_{\bm \phi^{\star}}\triangleq\frac{\partial {\bm\eta}^\star}{\partial\bm\phi^{\star}}\in\mathbb{R}^{2(K+1)\times2(K+1)}$ can be derived according to Equation \eqref{eq:geometry}. Let $\hat{\bm\phi}_{\text{Bob}}$ be Bob's estimate of $\bm\phi^{\star}$ with an \textit{unbiased estimator}. The mean squared error (MSE) with such an estimator can be bounded as follows \cite{Scharf}\vspace{-3pt}
\begin{equation}\label{eq:BobCRB}
    \mathbb{E}\left\{\left(\hat{\bm\phi}_{\text{Bob}}-{\bm\phi}^{\star}\right)\left(\hat{\bm\phi}_{\text{Bob}}-{\bm\phi}^{\star}\right)^{\mathrm{T}}\right\}\succeq\bm\Xi_{\bm \phi^{\star}}\triangleq{\bm J_{\bm \phi^{\star}}^{-1}},\vspace{-3pt}
\end{equation}
which is the well-known CRB. We note that our prior work on atomic norm minimization based localization yields estimators with bias that decrease as the SNR increases \cite{LOCMAN}.

\vspace{-6pt} 
\subsection{Eve's Localization Error}\label{subsec:eveerror}
To analyze Eve's localization accuracy, employing an \textit{efficient \rev{(CRB achieving)} estimator} for channel estimation, we denote by $\hat{\bm\eta}_{\text{Eve}}$ Eve’s estimate of $\bar{\bm\eta}$, and assume the true distribution of $\hat{\bm\eta}_{\text{Eve}}$ as in \cite{ZhengmismatchRIS}, \textit{i.e.,}\vspace{-3pt} 
\begin{equation}\label{eq:truemodel}
\hat{\bm\eta}_{\text{Eve}}=u(\bm\phi^\star)+\bm\epsilon,\vspace{-3pt} 
\end{equation}
where $\bm\epsilon$ is a zero-mean Gaussian random vector with covariance matrix $\bm\Sigma_{\bar{\bm\eta}}\triangleq\bm J_{\bar{\bm\eta}}^{-1}$. Herein, we have $\bar{\bm\eta}=u(\bm\phi^\star)$ with $u(\cdot)$ being a function mapping the location information $\bm\phi^\star$ to the location-relevant channel parameters $\bar{\bm\eta}$ according to the true geometric model defined in Equation \eqref{eq:mismatchgeometry}. However, since the shared information is unavailable to Eve, for a given vector of potential locations of Alice and scatterer $\bar{\bm\phi}\triangleq [(\bar{\boldsymbol{p}})^{\mathrm{T}},(\bar{\boldsymbol{v}}_1)^{\mathrm{T}},(\bar{\boldsymbol{v}}_2)^{\mathrm{T}},\cdots,(\bar{\boldsymbol{v}}_K)^{\mathrm{T}}]^{\mathrm{T}}\in\mathbb{R}^{2(K+1)\times1}$, Eve will incorrectly believe that the estimates of channel parameters $\hat{\bm\eta}_{\text{Eve}}$ are modelled as \vspace{-3pt} 
\begin{equation}\label{eq:mismatchmodel}
\hat{\bm\eta}_{\text{Eve}}=o(\bar{\bm\phi})+\bm\epsilon,\vspace{-3pt} 
\end{equation}
where $o(\cdot)$ is function similar to $u(\cdot)$, but is defined according to the mismatched geometric model assumed in Equation \eqref{eq:geometry}. The true and mismatched distributions of $\hat{\bm\eta}_{\text{Eve}}$ are denoted as $g_{\text{T}}(\hat{\bm\eta}_{\text{Eve}}|\bm\phi^\star)$ and $g_{\text{M}}(\hat{\bm\eta}_{\text{Eve}}|\bar{\bm\phi})$, respectively. 

Then, we can find a vector of the pseudo-true (incorrect) locations of Alice and scatterers $\bar{\bm\phi}^\star\triangleq [(\bar{\boldsymbol{p}}^\star)^{\mathrm{T}},(\bar{\boldsymbol{v}}^\star_1)^{\mathrm{T}},(\bar{\boldsymbol{v}}^\star_2)^{\mathrm{T}},\cdots,(\bar{\boldsymbol{v}}^\star_K)^{\mathrm{T}}]^{\mathrm{T}}\in\mathbb{R}^{2(K+1)\times1}$, when Eve exploits Equation \eqref{eq:mismatchmodel} for localization though the true distribution of $\hat{\bm\eta}_{\text{Eve}}$ is defined according to Equation \eqref{eq:truemodel}, such that \cite{Fortunatimismatchsurvey} \vspace{-3pt} 
\begin{equation}\label{eq:KLD}
    \bar{\bm\phi}^\star = \arg\min_{\bar{\bm\phi}} \operatorname{KL}\left(g_{\text{T}}(\hat{\bm\eta}_{\text{Eve}}|\bm\phi^\star)\|g_{\text{M}}(\hat{\bm\eta}_{\text{Eve}}|\bar{\bm\phi})\right),\vspace{-3pt} 
\end{equation}
where $\operatorname{KL}(\cdot\|\cdot)$ represents the Kullback–Leibler (KL) divergence for two given distributions. The closed-form expression of $\bar{\bm\phi}^\star$ will be derived later in this subsection. 

Denote by $\hat{\bm \phi}_\text{Eve}$ a \textit{misspecified-unbiased} estimator designed according to the mismatched model in Equation \eqref{eq:mismatchmodel}. With the DAIS for location-privacy enhancement, there is a lower bound for the MSE of Eve’s localization based on the analysis of the MCRB \cite{Fortunatimismatchsurvey} \vspace{-3pt} 
\begin{equation}\label{eq:EveLB}
\begin{aligned}
    &\mathbb{E}\left\{\left(\hat{\bm \phi}_\text{Eve}-{\bm \phi}^{\star}\right)\left(\hat{\bm \phi}_\text{Eve}-{\bm \phi}^{\star}\right)^{\mathrm{T}}\right\}\\
    &\quad\succeq \bm\Psi_{\bar{\bm \phi}^\star}\triangleq{\underbrace{\bm A_{\bar{\bm \phi}^\star}^{-1}\bm B_{\bar{\bm \phi}^\star}\bm A_{\bar{\bm \phi}^\star}^{-1}}_{\bm\Psi_{\bar{\bm \phi}^\star}^{(\romannumeral 1)}}+\underbrace{(\bar{\bm\phi}^\star-\bm\phi^{\star})(\bar{\bm\phi}^\star-\bm\phi^{\star})^\mathrm{T}}_{\bm\Psi_{\bar{\bm \phi}^\star}^{(\romannumeral 2)}}}, \vspace{-3pt} 
\end{aligned}
\end{equation}
where $\bm A_{\bar{\bm \phi}^\star}\in\mathbb{R}^{2(K+1)\times2(K+1)}$ and $\bm B_{\bar{\bm \phi}^\star}\in\mathbb{R}^{2(K+1)\times2(K+1)}$ are two generalized FIMs, defined as \vspace{-3pt} 
\begin{equation}\label{eq:A}
    \bm A_{\bar{\bm \phi}^\star}[r,l]\triangleq\mathbb{E}_{g_{\text{T}}(\hat{\bm\eta}_{\text{Eve}}|\bm\phi^\star)}\left\{\frac{\partial^2}{\partial\bar{\bm\phi}^\star[r]\partial\bar{\bm\phi^\star}[l]}\log g_{\text{M}}(\hat{\bm\eta}_{\text{Eve}}|\bar{\bm\phi}^\star)\right\},\vspace{-3pt} 
\end{equation}
and%\vspace{-3pt} 
\begin{equation}\label{eq:B}
\begin{aligned}
&\bm B_{\bar{\bm \phi}^\star}[r,l]\\
&\triangleq\mathbb{E}_{g_{\text{T}}(\hat{\bm\eta}_{\text{Eve}}|\bm\phi^\star)}\left\{\frac{\partial\log g_{\text{M}}(\hat{\bm\eta}_{\text{Eve}}|\bar{\bm\phi}^\star) }{\partial\bar{\bm\phi}^\star[r]}\frac{\partial\log g_{\text{M}}(\hat{\bm\eta}_{\text{Eve}}|\bar{\bm\phi}^\star) }{\partial\bar{\bm\phi}^\star[l]}\right\},\vspace{-5pt} 
\end{aligned}
\end{equation}

The next step is to derive the closed-form expression of $\bar{\bm\phi}^{\star}$, which is given by the following lemma.

\begin{lemma}\label{lemma:pseudotruepos}
Given the true and mismatched distributions of the estimated parameters $\hat{\bm\eta}_{\text{Eve}}$ in Equations \eqref{eq:truemodel} and \eqref{eq:mismatchmodel}, the pseudo-true locations of Alice and scatterers are given by 
\vspace{-3pt}
\begin{subequations}\label{eq:ptrue}
    \begin{align} 
        \bar{\boldsymbol{p}}^{\star} &= \boldsymbol{z} - c\bar\tau_{k_{\min}}[\cos({\bar\theta}_{\mathrm{Tx},k_{\min}}),\sin(\bar\theta_{\mathrm{Tx},k_{\min}})]^{\mathrm{T}},\\
        \bar{\boldsymbol{v}}^\star_{{k_{\min}}} &=\frac{1}{2}\bar{b}^{\star}_{0}[\cos(\bar{\theta}_{\mathrm{Tx},0}), \sin(\bar{\theta}_{\mathrm{Tx},0})]^{\mathrm{T}}+\bar{\boldsymbol{p}}^{\star},\\
        %\bar{v}^\star_{{k_{\min}},y} &= \frac{1}{2}\bar{b}^{\star}_{0}\sin(\bar{\theta}_{\mathrm{Tx},0})+\bar{p}^{\star}_{y}, \\
        \bar{\boldsymbol{v}}^\star_{k} &=\frac{1}{2}\bar{b}^{\star}_{k}[\cos(\bar{\theta}_{\mathrm{Tx},k}),\sin(\bar{\theta}_{\mathrm{Tx},k})]^{\mathrm{T}}+\bar{\boldsymbol{p}}^{\star}, \quad \text{if} \ k\neq{k_{\min}},%\\
        %\bar{v}^\star_{k,y} &= \frac{1}{2}\bar{b}^{\star}_{k}\sin(\bar{\theta}_{\mathrm{Tx},k})+\bar{p}^{\star}_{y}, \quad \text{if} \ k\neq{k_{\min}},
    \end{align} 
\end{subequations}
where $k_{\min}\triangleq\arg\min_k\bar{\tau}_k$
and \vspace{-3pt}
\begin{equation}
    \begin{aligned}    
        \bar{b}^{\star}_{k}=\frac{\left(c\bar\tau_k\right)^2-\left(z_x-\bar{p}^{\star}_{x}\right)^2-\left(z_y-\bar{p}^{\star}_{y}\right)^2}{c\bar\tau_k-\left(z_x-\bar{p}^{\star}_{x}\right)\cos(\bar{\theta}_{\mathrm{Tx},k})-\left(z_y-\bar{p}^{\star}_{y}\right)\sin(\bar{\theta}_{\mathrm{Tx},k})}, %\vspace{-3pt} 
    \end{aligned}%\vspace{-2pt}
\end{equation}
with $k=1,2,\cdots,K$.
\end{lemma}
\begin{proof}
    See Appendix \ref{sec:proofpseudotruepos}.
\end{proof}
We emphasize that due to the \textit{phase wrapping}, the estimate of the original LOS may not have the smallest delay;  however, when we perform localization, the path with the smallest shifted TOA is assumed to be the LOS path as in \cite{Shahmansoori,LOCMAN,Ayyalasomayajula}. 

\vspace{-6pt} 
\subsection{Degraded Localization Accuracy}\label{sec:degradationdais}
By comparing the lower bounds for Bob's and Eve's localization error, derived in Equations \eqref{eq:BobCRB} and \eqref{eq:EveLB}, respectively, we can show that the proposed DAIS strategy can effectively decrease Eve's localization accuracy as follows.

\begin{proposition}\label{prop:degradation} 
Suppose that Equations \eqref{eq:truemodel} and \eqref{eq:mismatchmodel} characterize the true and mismatched distributions of the estimated channel parameters $\hat{\bm\eta}_{\text{Eve}}$, then, there always exists a constant $\sigma_0$ such that $\operatorname{Tr}(\bm\Psi_{\bar{\bm \phi}^\star})\geq\operatorname{Tr}(\bm\Xi_{\bm \phi^{\star}})$  when $0\leq\sigma\leq\sigma_0$, where $\sigma$ represents the standard deviation of the Gaussian noise ${w}^{(g,n)}$ while $\bm\Xi_{\bm \phi^{\star}}$ and $\bm\Psi_{\bar{\bm \phi}^\star}$ are defined in Equations \eqref{eq:BobCRB} and $\eqref{eq:EveLB}$, respectively.
\end{proposition}
\begin{proof}
    See Appendix \ref{sec:proofdegradation}.
\end{proof}

According to Equations \eqref{eq:BobCRB} and \eqref{eq:EveLB}, Proposition \ref{prop:degradation} shows that Eve cannot estimate Alice's location more accurately than Bob if the value of $\sigma$ is small enough. {Considering that $\lim_{\sigma\downarrow0}\operatorname{Tr}(\bm\Psi_{\bar{\bm \phi}^\star}^{(\romannumeral 1)})=0$ also holds}, the {geometric mismatch} introduced by the proposed DAIS scheme, {which corresponds to the quantity $\operatorname{Tr}(\bm\Psi_{\bar{\bm \phi}^\star}^{(\romannumeral 2)})$,  is dominant} in the degradation of Eve's localization accuracy at a high SNR. We will theoretically show the impact of the choices of the design parameter $\bm\Delta$ on Eve’s localization accuracy in the next subsection.

\vspace{-10pt}
\subsection{Impact of Parameter Design}
\label{subsec:parameterdesign}
To further investigate the influence of the design parameters $\Delta_\tau$, and $\Delta_\theta$ on Eve’s localization accuracy, we first substitute the pseudo-true locations of Alice and scatterers in Equation \eqref{eq:ptrue} into Equation \eqref{eq:EveLB}, simplifying the MCRB for Eve's localization error as follows.
\begin{corollary}\label{coro:EveLBsimplifed}
Given the true and mismatched distributions of the estimated parameters $\hat{\bm \eta}_{\text{Eve}}$ in Equations \eqref{eq:truemodel} and \eqref{eq:mismatchmodel}, the MSE of Eve’s localization can be bounded as \vspace{-3pt}
\begin{equation}\label{eq:EveLBsimplified}
\begin{aligned}
    &\mathbb{E}\left\{\left(\hat{\bm \phi}_\text{Eve}-{\bm \phi}^{\star}\right)\left(\hat{\bm \phi}_\text{Eve}-{\bm \phi}^{\star}\right)^{\mathrm{T}}\right\}\\
    &\quad\quad\quad\quad\quad\succeq \bm\Psi_{\bar{\bm \phi}^\star}={{\bm J_{\bar{\bm \phi}^{\star}}^{-1}}+{(\bar{\bm\phi}^\star-\bm\phi^{\star})(\bar{\bm\phi}^\star-\bm\phi^{\star})^\mathrm{T}}}.\vspace{-3pt}
\end{aligned}
\end{equation}
\end{corollary}
\begin{proof}
    See Appendix \ref{sec:proofEveLBsimplifed}.
\end{proof}

\begin{remark}[Interpretation of the simplified MCRB]
    {The lower bound consists of two parts, {\it i.e.,} ${\bm\Psi_{\bar{\bm \phi}^\star}^{(\romannumeral 1)}}={\bm J_{\bar{\bm \phi}^{\star}}^{-1}}$ and ${\bm\Psi_{\bar{\bm \phi}^\star}^{(\romannumeral 2)}}= {(\bar{\bm\phi}^\star-\bm\phi^{\star})(\bar{\bm\phi}^\star-\bm\phi^{\star})^\mathrm{T}}$. Note that $\bm J_{\bar{\bm \phi}^{\star}}$ is the effective FIM for the estimation of ${\bar{\bm \phi}^{\star}}$. This indicates that using the proposed DAIS scheme we virtually move Alice and all the scatterers to certain incorrect locations (pseudo-true locations) ${\bar{\bm \phi}^\star}$ and Eve is misled to estimate such \textit{pseudo-true} locations; for the inference of Alice's true location, Eve has to address both the estimation error of ${\bar{\bm \phi}^\star}$, {\it i.e.}, ${\bm\Psi_{\bar{\bm \phi}^\star}^{(\romannumeral 1)}}$ and the associated geometric mismatch, {\it i.e.}, ${\bm\Psi_{\bar{\bm \phi}^\star}^{(\romannumeral 2)}}$. }
\end{remark}
\begin{remark}[Distinction from the analysis in \cite{li2023fpi}]
The lower bound on Eve's estimation error in \cite[Equation (22)]{li2023fpi} is derived from the analysis of the CRB in the presence of the FPI, since Eve incorrectly perceives more paths in her channel given the \textit{geometric feasibility} of the injected fake paths \cite{li2023fpi}. In contrast, we introduce {geometric mismatch} herein, thus the results in \cite{li2023fpi} cannot be directly used to characterize the performance degradation caused by DAIS as revealed in Equation \eqref{eq:EveLBsimplified}.
\end{remark}

As shown in Equation \eqref{eq:EveLBsimplified}, it is still complicated to theoretically analyze the MCRB to determine how the design parameters $\Delta_\tau$, and $\Delta_\theta$ affect Eve's localization error due to the structure of the effective FIM for the estimation of ${\bar{\bm \phi}^{\star}}$. Therefore, we will provide a closed-form expression for the simplified MCRB in Corollary \ref{coro:EveLBsimplifed} under a mild condition, suggesting the impact of the design parameters $\Delta_\tau$ and $\Delta_\theta$. 

\begin{proposition}\label{prop:lowerboundpos}
Denote by $\hat{\bm p}_\text{Eve}$ the position estimated by Eve. If the paths of mmWave MISO OFDM channels are orthogonal to each other in the following sense\footnote{{Paths are approximately orthogonal to each other for a large number of symbols and transmit antennas due to the low-scattering sparse nature of the mm-Wave channels \cite{li2023fpi,mmWaveorthogonality}. We will numerically validate the analysis for the design of $\bm\Delta$ in Section \ref{subsec:validation} with a practical number of symbols and transmit antennas; we do not impose path orthogonality in the numerical results which validate the DAIS designs.}}, \textit{i.e.}, \vspace{-3pt}
\begin{equation}\label{eq:orthopath}
    \frac{2}{\sigma^2}\sum_{n=0}^{N-1}\sum_{g=1}^{G}\mathfrak{R}\left\{\left(\frac{\partial\bar{ u}^{(g,n)}}{\partial \bar{\xi}_k}\right)^{*}\frac{\partial \bar{ u}^{(g,n)}}{\partial \bar{\xi}_{k^{\prime}}}\right\}=0, \quad k\neq k^{\prime}, \vspace{-3pt}
\end{equation}
where $\bar{\xi}_k\in\{\bar{\tau}_k,\bar{\theta}_{\mathrm{Tx},k},\mathfrak{R}\{\gamma^{\star}_k\},\mathfrak{I}\{\gamma^{\star}_k\}$ and $\bar{\xi}_{k^{\prime}}$ is defined similarly, with $k,k^{\prime}=0,1,\cdots,K$, in the presence of the proposed DAIS strategy, for any real $\psi>0$, there always exists a positive integer $\mathcal{G}$ such that when $G\geq\mathcal{G}$ the MSE for Eve's localization can be bounded as,\vspace{-3pt} 
%\begin{figure*}[hb]
\begin{equation}\label{eq:EveLBsimplify}
\begin{aligned}
    &\mathbb{E}\left\{\left\|\hat{\bm p}_\text{Eve}-{\bm p}^{\star}\right\|^2_2\right\}\\
    &\geq C_1+\frac{C_2\bar{\tau}_{k_{\min}}^2}{\cos^2(\bar{\theta}_{\mathrm{Tx},{k_{\min}}})}+\left\|\bar{\bm p}^\star-{\bm p}^{\star}\right\|^2_2,\\
    &=C_1+\frac{C_2({\tau}_{k_{\min}}+\Delta_{\tau})^2}{\cos^2
    \left(\arcsin\left({\left(\sin(\theta^{\star}_{\mathrm{Tx}, k})+\sin(\Delta_\theta)\right)_{\left(-1,1\right]}}\right)\right)}\\
    &+ \Big(z_x -p^\star_x - c(\tau_{k_{\min}}+\Delta_\tau)\\
    &\quad\quad\quad\quad \times\cos\left(\arcsin\left({\left(\sin(\theta^{\star}_{\mathrm{Tx}, k})+\sin(\Delta_\theta)\right)_{\left(-1,1\right]}}\right)\right)\Big)^2\\
    &+ \left(z_y -p^\star_y - c(\tau_{k_{\min}}+\Delta_\tau){\left(\sin(\theta^{\star}_{\mathrm{Tx}, k})+\sin(\Delta_\theta)\right)_{\left(-1,1\right]}}\right)^2\vspace{-3pt} 
\end{aligned}
\end{equation}
%\end{figure*}
with probability of $1$, where\vspace{-5pt} 
\begin{subequations}
    \begin{align}
        C_1&\triangleq\frac{3\sigma^2c^2NT_s^2}{2G|h_{k_{\min}}|^2\pi^2(N^2-1)}-\frac{\psi}{G},\\
        C_2&\triangleq\frac{3\sigma^2c^2\lambda_c^2}{2G|h_{k_{\min}}|^2\pi^2d^2N(N_t^2-1)}.\vspace{-5pt} 
    \end{align}
\end{subequations}
\end{proposition}
\begin{proof}
    See Appendix \ref{sec:prooflowerboundpos}.
\end{proof}

From Equation \eqref{eq:EveLBsimplify}, we can observe that, using the proposed DAIS strategy for location-privacy enhancement, the adjustment of the design parameters $\Delta_\tau$ and $\Delta_\theta$ can affect Eve's localization accuracy by changing the Alice's pseudo-true location and the associated location-relevant (shifted) channel parameters. Based on Proposition \ref{prop:lowerboundpos}, the impact of the choices of $\Delta_\tau$ and $\Delta_\theta$ are individually elaborated as follows. 

\begin{remark}[Impact of Angle Shift]\label{remark:impactangleshift}
    {To degrade Eve's localization accuracy with the proposed DAIS strategy, the desired value of $\Delta^{d}_\theta$ is given by\vspace{-5pt} 
    \begin{equation}\label{eq:designangle}
    \begin{aligned}
        &\Delta^{d}_\theta \in \mathcal{A}\triangleq\\
        &\left\{\theta|\theta=\arcsin\left(\left(1-\sin(\theta^{\star}_{\mathrm{Tx}, {k_{\min}}})\right)_{\left(-1,1\right]}\right)+2m\pi,m\in\mathbb{Z}\right\} \vspace{-5pt} 
    \end{aligned}
    \end{equation}
    such that $\cos^2(\bar{\theta}_{\mathrm{Tx},{k_{\min}}})$=0 and the lower bound derived in Equation \eqref{eq:EveLBsimplify} goes to infinity, rendering the estimation problem unstable. Under the above design of $\Delta_\theta$, the angle information is not learnable by Eve and thus she cannot accurately infer Alice's position.}
\end{remark}

\begin{remark}[Impact of Delay Shift]\label{remark:impactdelayshift}
    {According to Proposition \ref{prop:lowerboundpos}, the lower bound on the MSE for Eve's localization is positively related with $\bar{\tau}_{k_{\min}}^2$ and $\left\|\bar{\bm p}^\star-{\bm p}^{\star}\right\|^2_2$.
    Hence, given the knowledge of $\bm\eta^\star$ and $\Delta_\theta\notin\mathcal{A}$, we can solve the following problem for the optimal design of $\Delta_\tau$,\vspace{-5pt} 
    \begin{equation}\label{eq:designdelay}
        \Delta_\tau^{d} = \arg\max_{\Delta_\tau} \frac{C_2\bar{\tau}_{k_{\min}}^2}{\cos^2(\bar{\theta}_{\mathrm{Tx},{k_{\min}}})}+\left\|\bar{\bm p}^\star-{\bm p}^{\star}\right\|^2_2. \vspace{-5pt} 
    \end{equation}
    We note that $\left\|\bar{\bm p}^\star-{\bm p}^{\star}\right\|^2_2$ is the distance between Alice's true location and the {spoofed} location. Therefore, in terms of the location-privacy enhancement, if we maximally increase $\left\|\bar{\bm p}^\star-{\bm p}^{\star}\right\|^2_2$, then we virtually move Alice as far away from her true location as possible.}
\end{remark}

\begin{remark}[Knowledge of the CSI]\label{remark:DAISisCSIfree}
We emphasize that, to achieve DAIS, shifting delay and angles with the proposed precoding design in Equation \eqref{eq:daisbeamformer} does not rely on the CSI, while the optimal design is CSI-dependent as observed in Equations \eqref{eq:designangle} and \eqref{eq:designdelay}. However, even without the aid of  CSI, the reduction of Eve's localization accuracy can be guaranteed at high SNRs according to Proposition \ref{prop:degradation}, when Alice's true location $\bar{\bm p}^\star$ is different from her pseudo-true location ${\bm p}^{\star}$. We can easily have $\bar{\bm p}^\star\neq{\bm p}^{\star}$ by setting $\Delta_\tau\neq  mNT_s$ with $m\in\mathbb{Z}$ or $\sin(\Delta_\theta)\neq 0$.
\end{remark}

\vspace{-16pt} 
\section{Robustness of Delay-Angle Information Spoofing}\label{sec:DAISworsecase}
As shown in Section \ref{sec:analysis}, due to the precoding design in Equation \eqref{eq:daisbeamformer} for DAIS, Eve is misled by incorrect locations, \textit{i.e.,} the pseudo-true locations. However, Eve might realize that DAIS is being used for location-privacy enhancement or even have the knowledge of the structure of the proposed precoder. Therefore, in this section, to show the robustness of the proposed DAIS strategy, we analyze Eve’s estimation accuracy for the unknown delay-angle shifts, \textit{i.e.,} $\Delta_\tau$ and $\Delta_\theta$, in the case of precoder structure leakage, and then elaborate upon the challenges that Eve cannot overcome, even if Eve employs multiple antennas.

\vspace{-8pt} 
\subsection{Hardness of Delay-Angle Shifts Estimation} \label{subsec:strucleakage}
Provided that the structure of the precoder designed in Equation \eqref{eq:daisbeamformer} is leaked\footnote{We assume that Eve knows the definition of the precoder ${\bm \Phi}^{(n)}$, but does not have the knowledge of $\Delta_\tau$ and $\Delta_\theta$ in this case.}, Eve will  endeavor to infer $\bm\Delta$ to compensate for the obfuscation; however, this is a statistically hard estimation problem if $\bm\Delta$ is an unknown deterministic vector according to the following proposition. 

\begin{proposition}\label{prop:strucleakge} Assume that $\bm\Delta$ is an unknown deterministic vector. Let ${\bm\chi}\triangleq\left[({\bm\tau}^\star)^\mathrm{T},({\bm\theta}_{\mathrm{Tx}}^\star)^{\mathrm{T}},\mathfrak{R}\{(\bm\gamma^{\star})^\mathrm{T}\},\mathfrak{I}\{(\bm\gamma^{\star})^\mathrm{T}\},\bm\Delta^\mathrm{T}\right]^{\mathrm{T}}\in\mathbb{R}^{(4K+6)\times 1}$ and $\bm J_{{\bm \chi}}\in\mathbb{R}^{(4K+6)\times (4K+6)}$ be a vector of the unknown channel parameters and the associated FIM, respectively. Then, $\bm J_{{\bm \chi}}$ is a singular matrix.
\end{proposition}
\begin{proof}
    See Appendix \ref{sec:proofstrucleakge}.
\end{proof}

In contrast to \cite{li2023fpi} where Eve can compensate when she knows the structure of the designed precoder, herein, with DAIS, Eve still cannot distinguish the shifts introduced by the precoder from the true delay and angle information according to Proposition \ref{prop:strucleakge}. Thus, DAIS is robust to precoder leakage.  Our result also indicates that the delay-angle shifts can be preset according to a certain protocol between Alice and Bob, significantly reducing the overhead caused by the transmission of the delay-angle shifts via the secure channel.

\vspace{-8pt} 
\subsection{Challenges for a Multi-Antenna Eavesdropper}

The efficacy of the proposed DAIS strategy has been theoretically validated in Section \ref{sec:analysis}. However, \rev{after} realizing that the true TOAs and AODs are artificially shifted, Eve can endeavor to employ multi-antennas to estimate the angle-of-arrivals (AOA)  for the inference of Alice's location. However, without accurate estimates for TOAs and AODs, achieving high-accuracy localization with a single eavesdropper is still a challenging problem in a multipath environment. To accurately infer Alice's location, Eve has to address the following problems.
\begin{enumerate}
    \item[P1)] {\textbf{Identification of the LOS path}: Assume the path with the smallest delay to be the LOS path for localization as in \cite{Shahmansoori,LOCMAN,Ayyalasomayajula}. In the presence of the NLOS paths, due to the delay shift and phase wrapping, DAIS will misled Eve to incorrectly leverage the channel estimates associated with the $k_{\min}$-th path to infer Alice's location, according to the analysis in Section \ref{sec:analysis}. If $k_{\min}\neq0$, the original LOS path does not have the smallest shifted TOA among all the paths. The location of a scatterer will be shifted and incorrectly estimated, resulting in an erroneous location for Alice due to the geometric mismatch, as shown in Figure \ref{fig:sysmimo}.}
    \item[P2)] {\textbf{Localization without collaborators}: If there is only one receiver, the lack of either delay information or angle information will typically result in ambiguity for a location estimate \cite{zekavat2019handbook}. Recall that a single eavesdropper is assumed in this paper, while the TOAs and AODs are artificially shifted via precoding. Hence, precisely acquiring Alice's location is challenging for Eve. We note that the AOAs are not spoofed in the proposed DAIS design. However, even assuming that Eve obtains perfect knowledge of the AOA of the LOS path, from Figure \ref{fig:sysmimo}, we can still observe infinitely many possible locations for Alice solely based on this AOA. A better localization algorithm design is needed for Eve to combat DAIS when there are no collaborators.}

    \begin{figure}[t]
    \centering
    \includegraphics[scale=0.46]{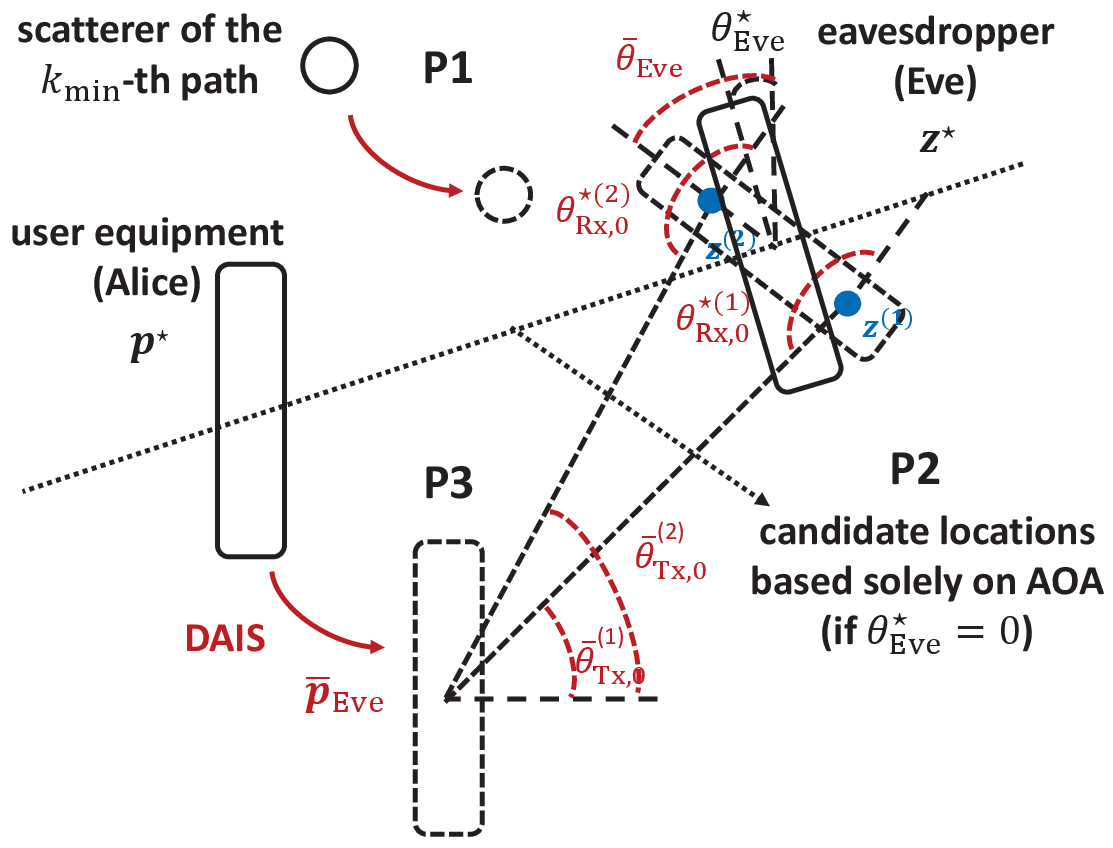}\vspace{-5.2pt}
    \caption{Depiction of orientation angle challenge for Eve when using DAIS in a multi-antenna eavesdropper.}\vspace{-15pt}
    \label{fig:sysmimo}
    \end{figure}
    \item[P3)]  {\textbf{Estimation of orientation angle of Eve's antenna array}: Suppose that Eve aims to accurately infer Alice's location based on AOA estimates, precise knowledge of the orientation angle of Eve's antenna array is also required. Denote by $\theta^{\star}_{\mathrm{Rx},0}$ and $\theta^{\star}_{\text{Eve}}$ the AOA of the LOS path and the true orientation angle of Eve's antenna array, respectively. According to \cite{LOCMAN,Shahmansoori}, without the location-privacy preservation mechanisms, the orientation angle can be precisely estimated for a MIMO system given good quality estimates for the TOA, AOD, and AOA, via the geometry as follows,\vspace{-5pt} 
    \begin{equation}\label{eq:orientation}
        {\theta}_{\text{Eve}}  =  \pi + {{\theta}}^{\star}_{\mathrm{Tx},0} - {\theta}^{\star}_{\mathrm{Rx},0},\vspace{-5pt} 
    \end{equation}
    However, using the proposed DAIS strategy, accurate TOA and AOD estimates are unavailable for orientation angle estimation and Eve has to resolve the ambiguity caused by the unknown orientation angle. }
    \par{To further illustrate the hardness of accurate localization without the knowledge of $\theta^{\star}_{\text{Eve}}$, we consider a relatively benign situation for Eve, where only the LOS path exists and Eve has perfect knowledge of channel parameters, \textit{i.e.,} AOA, shifted AOD, and shifted TOA. In such a case, Eve does not have to address problem P1. Without the aid of collaborators, we also assume that Eve adopts a sub-array approach as in \cite{subarray1,subarray2} to estimate Alice's location based solely on AOAs, where Eve’s antenna array can be divided into several non-overlapping sub-arrays and each sub-array can provide a distinct AOA estimate for localization. \rev{Eve is still in the far-field region for each sub-array (see \cite{subarray2}) and Equation \eqref{eq:channel_subcarrier} still holds.} As the intersection of at least two lines of bearing (LOB) will provide an location estimate \cite{zekavat2019handbook}, for simplicity, as shown in Figure \ref{fig:sysmimo}, Eve's antenna array is divided into two sub-arrays and the associated shifted AODs and AOAs are represented as $\bar{\theta}^{(i)}_{\mathrm{Tx},0}$ and $\theta^{\star(i)}_{\mathrm{Rx},0}$, with $i=1,2$.} Given a value for the orientation angle of Eve's antenna array $\bar{\theta}_{\text{Eve}}$, the locations of these two sub-arrays, denoted as $\boldsymbol{z}^{(1)}=\left[z^{(1)}_{x},z^{(1)}_{y}\right]^{\mathrm{T}}\in\mathbb{R}^{2\times1}$ and $\boldsymbol{z}^{(2)}=\left[z^{(2)}_{x},z^{(2)}_{y}\right]^{\mathrm{T}}\in\mathbb{R}^{2\times1}$, can be derived based on Eve’s location $\boldsymbol{z}^{\star}$,\vspace{-5pt} 
\begin{equation}\label{eq:subarrayloc}
\begin{aligned}
    \boldsymbol{z}^{(1)} &= \left[\boldsymbol{z}^{\star}_x+\frac{D}{4}\sin\left(\bar{\theta}_{\text{Eve}}\right), \boldsymbol{z}^{\star}_y-\frac{D}{4}\cos\left(\bar{\theta}_{\text{Eve}}\right)\right]^{\mathrm{T}},\\
    \boldsymbol{z}^{(2)} &= \left[\boldsymbol{z}^{\star}_x-\frac{D}{4}\sin\left(\bar{\theta}_{\text{Eve}}\right), \boldsymbol{z}^{\star}_y+\frac{D}{4}\cos\left(\bar{\theta}_{\text{Eve}}\right)\right]^{\mathrm{T}},\vspace{-5pt} 
\end{aligned}   
\end{equation}
while the two LOBs are characterized by\vspace{-5pt} 
\begin{equation}\label{eq:lobs}
    \begin{aligned}
        &\bar{{p}}^{(1)}_y = \tan\left({\theta}^{\star(1)}_{\mathrm{Rx},0}+\bar{\theta}_{\text{Eve}}\right)\left(\bar{{p}}^{(1)}_x - {z}^{(1)}_x\right)+{z}^{(1)}_y,\\
        & \bar{{p}}^{(2)}_y = \tan\left({\theta}^{\star(2)}_{\mathrm{Rx},0}+\bar{\theta}_{\text{Eve}}\right)\left(\bar{{p}}^{(2)}_x - {z}^{(2)}_x\right)+{z}^{(2)}_y.\vspace{-5pt} 
    \end{aligned}
\end{equation}
For a given value of $\bar{\theta}_{\text{Eve}}$, $\bar{\boldsymbol{p}}^{(1)}=\left[\bar{p}^{(1)}_{x},\bar{p}^{(1)}_{y}\right]^{\mathrm{T}}\in\mathbb{R}^{2\times1}$ and $\bar{\boldsymbol{p}}^{(2)}=\left[\bar{p}^{(2)}_{x},\bar{p}^{(2)}_{y}\right]^{\mathrm{T}}\in\mathbb{R}^{2\times1}$ represent the locations estimated based on ${\theta}^{\star(1)}_{\mathrm{Rx},0}$ and ${\theta}^{\star(2)}_{\mathrm{Rx},0}$, respectively. %With slightly abuse of notation, 
We use $\bar{\boldsymbol{p}}=\left[\bar{p}_{x},\bar{p}_{y}\right]^{\mathrm{T}}\in\mathbb{R}^{2\times1}$ for Eve’s incorrect perception of Alice’s location based on the AOAs and orientation angle. By solving the system of linear equations in Equation \eqref{eq:lobs} via setting $\bar{\boldsymbol{p}}^{(1)} = \bar{\boldsymbol{p}}^{(2)}$, and substituting Equation \eqref{eq:subarrayloc} into Equation \eqref{eq:lobs}, such a perceived location can be expressed as\vspace{-5pt} 
\begin{equation}\label{eq:hatpaoa}
\begin{aligned}
    &\bar{{p}}_{x}=  \\
    &z_x-\frac{D\sin\left(\bar{\theta}_{\text{Eve}}\right)\sum_{i=1}^2\tan\left({\theta}^{\star(i)}_{\mathrm{Rx},0}+\bar{\theta}_{\text{Eve}}\right)}{4\left(\tan\left({\theta}^{\star(2)}_{\mathrm{Rx},0}+\bar{\theta}_{\text{Eve}}\right)-\tan\left({\theta}^{\star(1)}_{\mathrm{Rx},0}+\bar{\theta}_{\text{Eve}}\right)\right)} \\
    &- \frac{D\cos\left(\bar{\theta}_{\text{Eve}}\right)}{2\left(\tan\left({\theta}^{\star(2)}_{\mathrm{Rx},0}+\bar{\theta}_{\text{Eve}}\right)-\tan\left({\theta}^{\star(1)}_{\mathrm{Rx},0}+\bar{\theta}_{\text{Eve}}\right)\right)},\\
    &\bar{{p}}_{y}= \\
    &z_y-\frac{D\cos\left(\bar{\theta}_{\text{Eve}}\right)\sum_{i=1}^2\tan\left({\theta}^{\star(2)}_{\mathrm{Rx},0}+\bar{\theta}_{\text{Eve}}\right)}{4\left(\tan\left({\theta}^{\star(2)}_{\mathrm{Rx},0}+\bar{\theta}_{\text{Eve}}\right)-\tan\left({\theta}^{\star(1)}_{\mathrm{Rx},0}+\bar{\theta}_{\text{Eve}}\right)\right)}\\
    & - \frac{D\sin\left(\bar{\theta}_{\text{Eve}}\right)\tan\left({\theta}^{\star(1)}_{\mathrm{Rx},0}+\bar{\theta}_{\text{Eve}}\right)\tan\left({\theta}^{\star(1)}_{\mathrm{Rx},0}+\bar{\theta}_{\text{Eve}}\right)}{2\left(\tan\left({\theta}^{\star(2)}_{\mathrm{Rx},0}+\bar{\theta}_{\text{Eve}}\right)-\tan\left({\theta}^{\star(1)}_{\mathrm{Rx},0}+\bar{\theta}_{\text{Eve}}\right)\right)}.\vspace{-5pt} 
\end{aligned}    
\end{equation}
From Equation \eqref{eq:hatpaoa}, we can observe that Eve's perception of Alice's location is affected by the orientation angle of Eve's antenna array. If Eve still leverages Equation \eqref{eq:orientation} to acquire her orientation angle, \textit{e.g.}\footnote{For a small-sized antenna array, $ {\bar{\theta}}^{(1)}_{\mathrm{Tx},0} - {\theta}^{\star(1)}_{\mathrm{Rx},0}\approx{\bar{\theta}}^{(2)}_{\mathrm{Tx},0} - {\theta}^{\star(2)}_{\mathrm{Rx},0}$ holds.},\vspace{-5pt} 
\begin{equation}
    \bar{\theta}_{\text{Eve}}  =  \pi + {\bar{\theta}}^{(1)}_{\mathrm{Tx},0} - {\theta}^{\star(1)}_{\mathrm{Rx},0},\vspace{-5pt} 
\end{equation}
under the perturbation of the angle shift as in Equation \eqref{eq:mismatchgeometry}, the location perceived by Eve can be far away from the Alice's true location, as illustrated in Figure \ref{fig:sysmimo}. 
In the presence of DAIS, the distance between Alice's true location $\boldsymbol{p}^\star$ and Eve's incorrect perception $\bar{\boldsymbol{p}}_{\text{Eve}}$ is plotted in Figure \ref{fig:incorrectperceptionMIMO}, where we assume that Eve is equipped with $16$ receive antennas and no NLOS paths are considered herein. The {values for the other signal parameters} are the same as that will be given in Section \ref{subsec:signalparams}. As observed in Figure \ref{fig:incorrectperceptionMIMO}, even in such a benign situation, the deviation of Eve's incorrect perception of Alice's location from the true one can be greater than $1$ m when $|\Delta_\theta|>0.1$ is used for DAIS\footnote{In contrast, angle information is not spoofed in \cite{Ayyalasomayajula} so Eve can estimate AOAs, AODs, as well as orientation angle as in Equation \eqref{eq:orientation}.  Thus, {for the same situation}, inferring Alice's location solely based on angle information will still yield location-privacy leakage in \cite{Ayyalasomayajula}.}.
Taking the multipath effects and the error in the channel estimates into consideration, we note that Eve's inference of Alice's location is even more challenging in practice. Extending DAIS to MIMO systems and collaborative eavesdroppers is a topic for future research.

    \begin{figure}[t]
    \centering
    \includegraphics[scale=0.52]{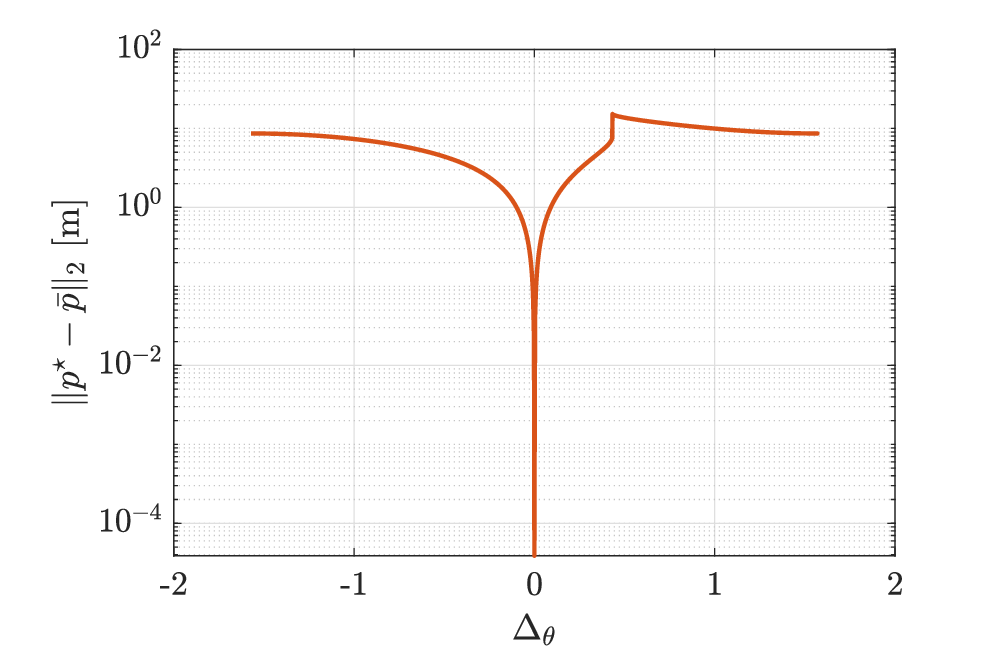}\vspace{-6pt}
    \caption{Deviations of the location perceived by a multi-antenna eavesdropper from the UE's true location, resulting from DAIS.} \vspace{-15pt}
    \label{fig:incorrectperceptionMIMO}
    \end{figure}
    
\end{enumerate}

\vspace{-9pt} 
\section{Simulation Results}\label{sec:sim}
In this section, to show that the proposed DAIS scheme effectively protects UE's location from being accurately estimated by an eavesdropper, we numerically evaluate the lower bounds on eavesdropper's localization error derived in Section \ref{subsec:eveerror}, which is the best performance that the eavesdropper can achieve with a misspecified-unbiased estimator. To be more specific, the theoretical analyses for the impact of parameter design presented in Section \ref{subsec:parameterdesign} are numerically validated. Then, the localization error achieved with the proposed DAIS strategy is measured against that without any location-privacy preservation to show the enhanced location privacy. In addition, comparisons with a CSI-free FPI scheme \cite{li2023fpi} and a CSI-dependent beamforming design \cite{Ayyalasomayajula} are also provided to validate the efficacy of the proposed DAIS. \rev{Finally, the accuracy degradation for a multi-antenna eavesdropper is investigated as well to show the robustness of the proposed strategy.}

\vspace{-6pt} 
\subsection{Signal Parameters}\label{subsec:signalparams}
In all of the numerical results, we set the parameters $K$, $\varphi_c$, $B$, $c$, $G$, $N_t$, $N$ to $2$, $60$ GHz, $30$ MHz, $300$ \rev{m/$\mu$s}, $16$, $16$, and $16$, respectively. For the channel model adopted in \eqref{eq:channel_subcarrier}, channel coefficients are numerically generated according to the free-space path loss model \cite{Goldsmith} while we place the scatterers of the two NLOS paths at $[7.44 \text{ m}, 8.53 \text{ m}]^{\mathrm{T}}$ and $[8.87\text{ m}, -6.05 \text{ m}]^{\mathrm{T}}$, respectively. Alice is at $[3 \text{ m},0 \text{ m}]^{\mathrm{T}}$ and the transmitted pilot signals are certain random, complex values uniformly generated on the unit circle.  For a fair comparison, Bob and Eve are placed at the same location $[10 \text{ m},5 \text{ m}]^{\mathrm{T}}$ so that the same received signals are leveraged for their individual localization. The RMSE for Bob's and Eve’s localization are defined as\vspace{-3pt} 
\begin{equation}
\begin{aligned} \operatorname{RMSE}_{\text{Bob}}&\triangleq\sqrt{\bm\Xi_{\bm \phi^{\star}}[1,1]+\bm\Xi_{\bm \phi^{\star}}[2,2]},\\    \operatorname{RMSE}_{\text{Eve}}&\triangleq\sqrt{\bm\Psi_{\bar{\bm \phi}^\star}[1,1]+\bm\Psi_{\bar{\bm \phi}^\star}[2,2]},  \vspace{-3pt}
\end{aligned}
\end{equation}
respectively. Unless otherwise stated, the SNR is defined as $\operatorname{SNR}\triangleq10\log_{10}\frac{\sum^{G}_{g=1}\sum^{N-1}_{n=0}|{\bar{u}}^{(g,n)}|^{2}}{NG\sigma^2}$.

\begin{figure}[t]
\centering
\includegraphics[scale=0.52]{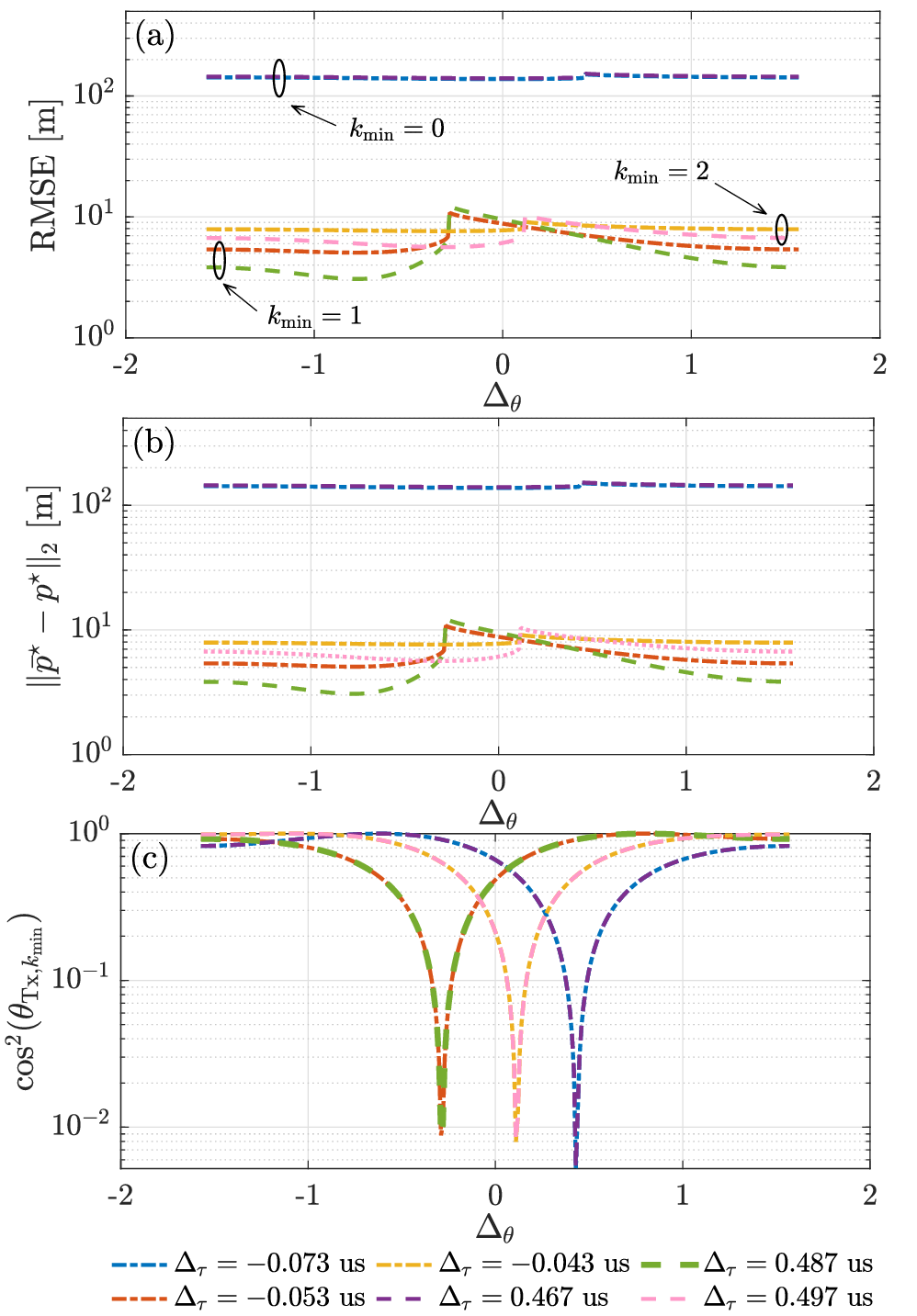}\vspace{-5pt}
\caption{{(a) Lower bounds for the RMSE of Eve’s localization with different choices of $\bm\Delta$, where $k_{\min}=k$ means that the $k$-th path with the smallest shifted TOA is assumed to be the LOS path for localization due to phase wrapping according to Equation \eqref{eq:ptrue}; (b) The associated values of $\|\bar{\bm p}^{\star}-\bar{\bm p}\|$; (c) The associated values of $\cos^2(\bar{\theta}_{\mathrm{Tx},k_{\min}})$.}}
\label{fig:parameterdesign} \vspace{-12pt}
\end{figure}

\begin{figure}[t]
\centering
\includegraphics[scale=0.5]{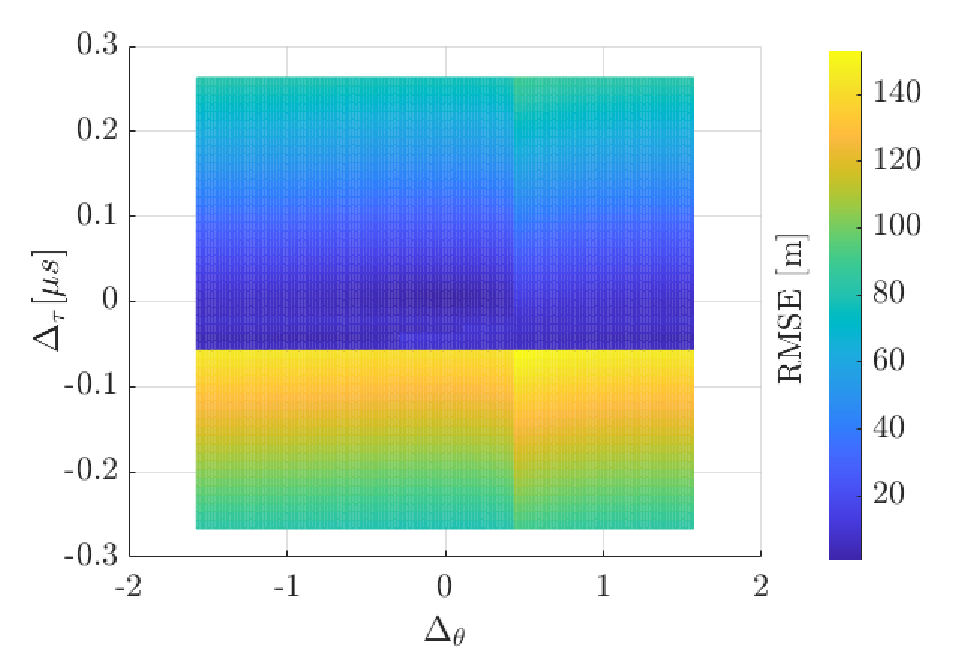}\vspace{-5pt}
\caption{Lower bounds for the RMSE of Eve’s localization with more choices of $\bm\Delta$.}
\label{fig:parameterdesignmorechoices} \vspace{-15pt}
\end{figure}

\vspace{-6pt} 
\subsection{Validation of Theoretical Analyses in Section \ref{subsec:parameterdesign}}\label{subsec:validation}
Without loss of generality, we assume the path produced by the scatterer at $[8.87\text{ m}, -6.05 \text{ m}]^{\mathrm{T}}$ is the first NLOS path, \textit{i.e.}, $k=1$, and the other is the second NLOS path, \textit{i.e.}, $k=2$. The numerical results in this subsection are generated with SNR being $20$ dB. Under different choices of $\Delta_\tau$ and $\Delta_\theta$, the corresponding RMSEs for Eve's localization accuracy are provided in Figure \ref{fig:parameterdesign} (a). From Figure \ref{fig:parameterdesign} (a), it can be observed that variations of the RMSE are more pronounced when value of $\Delta_\tau$ is changed. This indicates that the adjustment of $\Delta_\theta$ has a relatively less influence on Eve’s performance degradation provided that $\cos^2(\bar{\theta}_{\mathrm{Tx},k_{\min}})$ is not close to zero. {Based on Figure \ref{fig:parameterdesign} (a), we underscore the non-linear and non-monotonic effect of changing the delay-angle shifts on the localization accuracy; hence the need for the analysis to understand the impact of these design parameters. 

Furthermore, for the given choices of $\Delta_\tau$ and $\Delta_\theta$, the associated values for $\|\bar{\bm p}^{\star}-\bar{\bm p}\|$ and $\cos^2(\bar{\theta}_{\mathrm{Tx},k_{\min}})$ are shown in Figure \ref{fig:parameterdesign} (b) and (c), respectively. We note that, although the orthogonality among the paths does not strictly hold in the numerical results, relatively higher RMSEs are still achieved at large values of $\|\bar{\bm p}^{\star}-\bar{\bm p}\|$ and small values of $\cos^2(\bar{\theta}_{\mathrm{Tx},k_{\min}})$ according to Figure \ref{fig:parameterdesign} (b) and (c), which coincides with the theoretical analyses of the parameter designs in Remarks \ref{remark:impactangleshift} and \ref{remark:impactdelayshift}. Hence, to significantly increase Eve's localization error with the proposed DAIS strategy, one option is to increase the distance between Alice's true and pseudo-true locations as much as possible via the design of the delay-angle shifts; under the knowledge of $k_{\min}$ and ${\theta}^{\star}_{\mathrm{Tx},{k_{\min}}}$, another option is to adjust $\Delta_\theta$ such that the line between Alice's pseudo-true location and Eve's location is parallel to Alice's antenna array.

Although the optimal design of parameters $\Delta_\tau$ and $\Delta_\theta$ relies on the CSI, the proposed DAIS strategy still can enhance location privacy without  CSI, as analyzed in Remark \ref{remark:DAISisCSIfree}.  The lower bounds for Eve's localization error with more choices of $\boldsymbol{\Delta}$ are provided in Figure \ref{fig:parameterdesignmorechoices}. The value of $\Delta_\tau$ is changed from $-\frac{NT_s}{2}$ to $\frac{NT_s}{2}$ with the resolution of $0.01$ $\mu s$ while $\Delta_\theta$ ranges from $-\frac{\pi}{2}$ to $\frac{\pi}{2}$ with the resolution being 0.01. Consistent with the analysis in Remark \ref{remark:DAISisCSIfree}, when the values of $\Delta_\tau$ and $\Delta_\theta$ deviate from zeros, Eve's localization accuracy is strongly degraded as shown in Figure \ref{fig:parameterdesignmorechoices}. We underscore that the values of $N$ and $B$ can strongly influence these error bounds as well. {In the absence of CSI}, we can randomly and uniformly generate the delay and angle shift over $\left[-\frac{NT_s}{2},\frac{NT_s}{2}\right]$ and $\left[-\frac{\pi}{2},\frac{\pi}{2}\right]$, respectively. {The numerical results averaged over these channel parameters will be provided in the sequel.}

\vspace{-9pt} 
\subsection{Estimation Accuracy Comparison}\label{subsec:accuracycomparison}

\begin{figure}[t]
\centering
\includegraphics[scale=0.52]{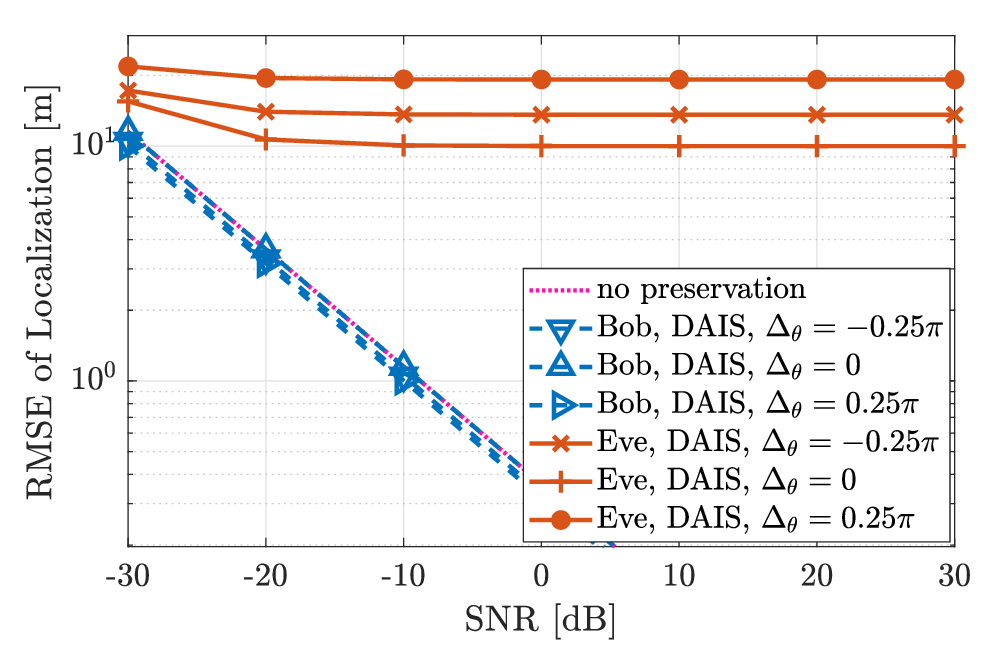}\vspace{-10pt}
\caption{{Lower bounds for the RMSE of localization with different choices of $\Delta_\theta$, where $\Delta_\tau=T_s$.}}
\label{fig:rmse_angle} \vspace{-12pt}
\end{figure}
\begin{figure}[t]
\centering
\includegraphics[scale=0.52]{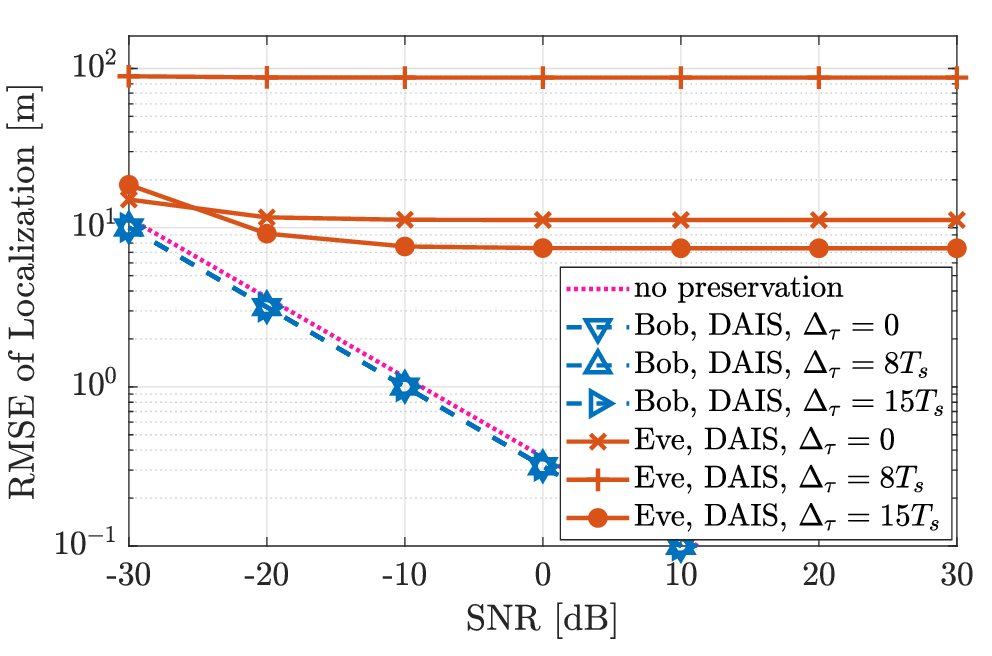}\vspace{-6pt}
\caption{{Lower bounds for the RMSE of localization with different choices of $\Delta_\tau$, where $\Delta_\theta=0.25\pi$.}}
\label{fig:rmse_delay}\vspace{-15pt}
\end{figure}

Using the proposed DAIS strategy for location-privacy enhancement, the RMSEs for Bob's and Eve's localization are provided in Figure \ref{fig:rmse_angle}, where $\Delta_\tau$ is fixed at $T_s$ and $\Delta_\theta$ is set to $-0.25\pi$, $0$, and $0.25\pi$, respectively. Since the delay-angle shifts are securely shared with Bob such that he can construct the effective pilot signal $\bar{\bm s}^{(g,n)}$ for localization, the geometric mismatch caused by the proposed DAIS strategy can be removed and as compared with the case without any location-privacy preservation mechanisms, negligible loss 
is observed for Bob in Figure \ref{fig:rmse_angle}. In contrast, the lack of the knowledge of $\Delta_\tau$ and $\Delta_\theta$ results in Eve's incorrect perception of Alice's location, {yielding a strong degradation of Eve's localization accuracy as shown in Figure \ref{fig:rmse_angle}}. To be more specific, when $\operatorname{SNR}$ and $\Delta_\theta$ is set to $0$ dB and $0.25\pi$, respectively, $\operatorname{RMSE}_{\text{Bob}}$ can be maintained at around $0.32$ m by virtue of the shared information, while obfuscation introduced by DAIS increases $\operatorname{RMSE}_{\text{Eve}}$ to around $19.22$ m. Consistent with the analysis in Section \ref{sec:degradationdais}, at high SNRs, Eve’s localization error is mainly affected by the distance between Alice's true location and the associated pseudo-true location; such distances are around $19.22$ m, $10.00$ m, and $13.61$ m for $\Delta_\theta=0.25\pi$, $\Delta_\theta=0$, and $\Delta_\theta=-0.25\pi$, respectively, which leads to distinct degradation for Eve's localization.

To strengthen the location-privacy enhancement with the proposed DAIS strategy, the delay shift $\Delta_\tau$ can be also adjusted, moving Alice’s pseudo-true location further away from the true location. As shown in Figure \ref{fig:rmse_delay}, for a given $\Delta_\theta=0.25\pi$, using DAIS, we can increase Eve’s localization error to $87.66$ m at $\operatorname{SNR}=0$ dB if $\Delta_\tau=8T_s$. Similar to the adjustment of $\Delta_\theta$, the lower bound on Eve's localization error is not monotonically increasing with respect to $\Delta_\theta$ due to phase wrapping. We observe that the largest value of $\Delta_\tau$ in Figure \ref{fig:rmse_delay}, \textit{i.e.,} $\Delta_\tau=15T_s$, does not result in the worst localization accuracy for Eve, where a NLOS path is incorrectly treated as the LOS path for localization, yet there is more than a $15$ dB gap as compared with Bob who knows the shared information when $\operatorname{SNR}$ is higher than $-10$ dB\footnote{For low SNRs, the effect of the geometric mismatch is relatively less significant because of the noise. The exact performance is also affected by other factors, \textit{e.g.,} AODs $\bar{\bm\theta}_{\mathrm{Tx}}$, according to the analysis in \cite{li2023fpi}.}.

\begin{figure}[t]
\centering
\includegraphics[scale=0.52]{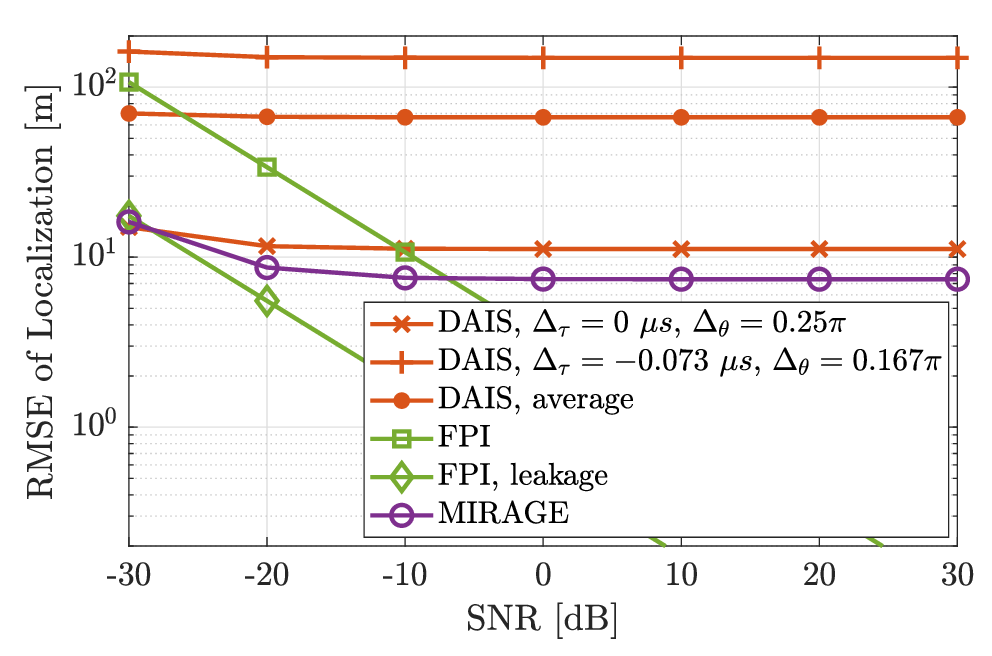}\vspace{-10pt}
\caption{Comparisons of lower bounds for Eve's localization accuracy with different location-privacy enhancement strategies.}
\label{fig:comparisons}\vspace{-15pt}
\end{figure}

Furthermore, the lower bounds for Eve's localization error with the CSI-free FPI derived in \cite{li2023fpi} are shown in Figure \ref{fig:comparisons} as a comparison, where the design parameters $\bar\delta_\tau$ and $\bar\delta_{\theta_{\text{TX}}}$ used for the FPI scheme are set according to \cite{li2023fpi}. Following the analyses in Sections \ref{subsec:parameterdesign} and \ref{subsec:validation}, with a good choice of design parameter for DAIS, \textit{e.g.,} $\Delta\tau=-0.073$ $\mu s, \Delta_\theta=0.167\pi$ in contrast to $\Delta\tau=0$ $\mu s, \Delta_\theta=0.25\pi$, the RMSE exceeding $150$ m for Eve’s localization can be seen in Figure \ref{fig:comparisons}, which outperforms the FPI strategy in the considered SNR regime. To show the efficacy of the proposed DAIS strategy in the absence of the CSI, we also evaluate the Eve's localization accuracy in an average sense, where $\Delta_\tau$ and $\Delta_\theta$ are randomly and uniformly generated over $\left[-\frac{NT_s}{2},\frac{NT_s}{2}\right]$ and $\left[-\frac{\pi}{2},\frac{\pi}{2}\right]$, respectively, and the corresponding lower bounds are averaged over $1000$ realizations of $\boldsymbol{\Delta}$. As compared with the FPI scheme \cite{li2023fpi}, the proposed DAIS design results in a comparable accuracy degradation for Eve at low SNRs while it is more effective at high SNRs due to the introduced geometric mismatch. In addition, \rev{as shown in Figure \ref{fig:comparisons}, if the structure of the precoder is leaked to Eve, the efficacy of the FPI scheme \cite{li2023fpi} will be degraded; in contrast, the proposed DAIS strategy is less sensitive to such a leakage based on the analysis in Section \ref{subsec:strucleakage}.}

\begin{figure}[t]
\centering
\includegraphics[scale=0.52]{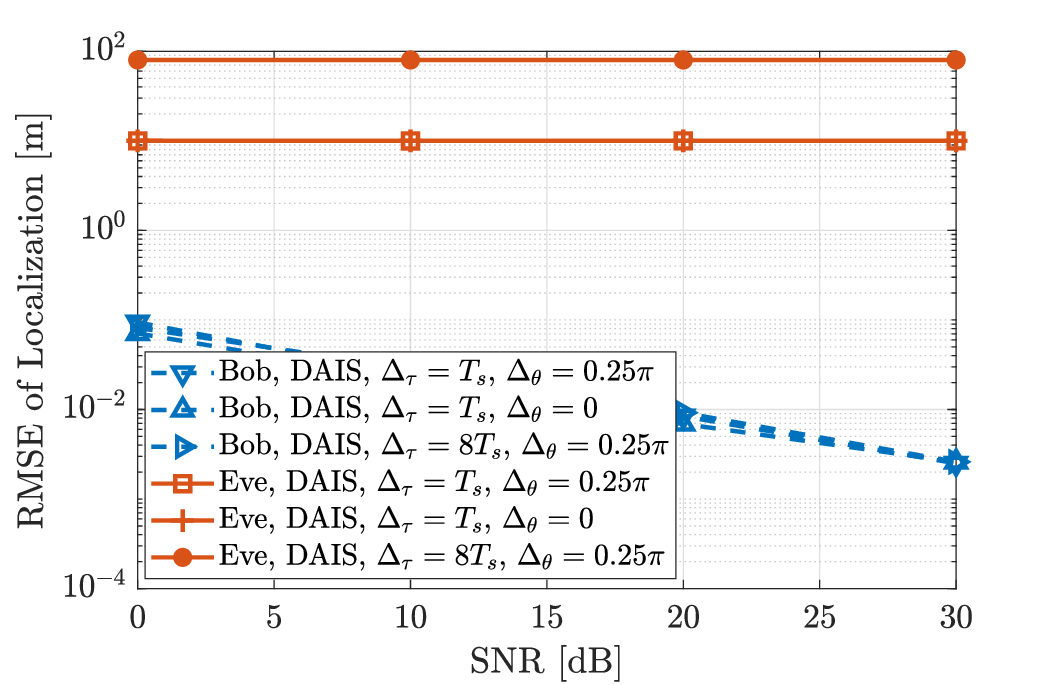}\vspace{-10pt}
\caption{\rev{Comparisons of RMSE of localization in the presence of unknown orientation angle: multi-antenna Bob versus multi-antenna Eve.}}\label{fig:DAIS-MIMO}\vspace{-15pt}
\end{figure}

We note that, from the aspect of obfuscating the LOS path, the proposed DAIS scheme can be considered as an extension of \cite{Ayyalasomayajula} since DAIS can also results in one NLOS path incorrectly treated as the LOS path for localization when $k_{\min}\neq0$. However, the design in \cite{Ayyalasomayajula} relies on the CSI to add an extra delay to the LOS path via a beamforming design such that one of NLOS paths has the smallest TOA. Denote by $\tau_{\text{NLOS},\min}\triangleq\min_{k=1,\cdots,K} \tau_k$ the smallest TOA of all the NLOS paths, and $\tau_{\text{obf}}$ as the extra delay added to the LOS path. Based on perfect CSI, the TOA of the LOS path is shifted by $\tau_{\text{obf}}$ via a beamforming design in \cite{Ayyalasomayajula} and the corresponding MCRB for Eve's localization error is also provided in Figure \ref{fig:comparisons}, which can be similarly derived according to Section \ref{subsec:eveerror}. Herein, we set $\tau_{\text{obf}}$ to $0.0279$ $\mu s$ so $\tau_0+\tau_{\text{obf}} > \tau_{\text{NLOS},\min}$ holds for the system configuration given in Section \ref{subsec:signalparams}. From Figure \ref{fig:comparisons}, we can still observe a stronger degradation for Eve's localization accuracy with the proposed DAIS scheme, even though perfect CSI is leveraged in the design of MIRAGE \cite{Ayyalasomayajula}.

\rev{For notational clarity, our DAIS design is derived for known orientation angle for Alice and single antenna for Bob and Eve; however, the design works without these assumptions. To show the robustness of our proposed strategy, we consider multiple antennas for both Bob and Eve as in \cite{LOCMAN}, where Alice's orientation is unknown. We assume Alice uses the same precoder proposed in Section \uppercase\expandafter{\romannumeral3}-B for DAIS, while both Bob and Eve leverage the super-resolution algorithm in \cite{LOCMAN} to infer Alice's location. The method of \cite{LOCMAN} was shown to nearly achieve the CRB. As observed in Figure \ref{fig:DAIS-MIMO}, the multi-antenna Eve still cannot infer Alice's location accurately in contrast to the multi-antenna Bob who has the shared delay-angle shifts.}

\vspace{-9pt}
\section{Conclusions}\label{sec:con}
In this paper, we provided a novel DAIS strategy which strongly enhanced location privacy against the eavesdropper. The DAIS strategy in turns, spoofed the location of the transmitting UE. Our new strategy did not require CSI between any of the transmitters or receivers.}
{Our DAIS method was implemented via a precoder}, which was designed based on the intrinsic structure of the wireless channel.  The legitimate localizer was still able to infer UE's location via the secure sharing of only two parameters.
A lower bound on eavesdropper's localization error was derived in the presence of DAIS, theoretically validating that  the eavesdropper can be misled to estimate an incorrect location as the true one. Assuming the orthogonality among the paths, the lower bound was simplified  to yield an explicit function of the delay-angle shifts, showing the sensitivity to the various shift parameters and suggesting the appropriate values of these design parameters. Furthermore, the robustness of DAIS against  precoder structure leakage as well as the presence of a multi-antenna eavesdropper was analyzed. Using the proposed DAIS strategy for location privacy enhancement, there was more than $15$ dB accuracy degradation for the eavesdropper in contrast to the legitimate localizer's estimation error at high SNRs, and the RMSE for the eavesdropper's localization can be up to 150 m.  Our proposed strategy was numerically shown to offer  strong performance gains over state of the art strategies: our own CSI-free FPI strategy \cite{li2023fpi} and the CSI-dependent MIRAGE design \cite{Ayyalasomayajula}.

\appendices
\begingroup
\allowdisplaybreaks

\vspace{-8pt} 
\section{Proof of Lemma \ref{lemma:pseudotruepos}}\label{sec:proofpseudotruepos}\vspace{-3pt}
Since $\bar{\bm\phi}^{\star}$ is the parameter vector that minimizes the KL divergence for two given Gaussian distributions $g_{\text{T}}(\hat{\bm\eta}_{\text{Eve}}|\bm\phi^\star)$ and $g_{\text{M}}(\hat{\bm\eta}_{\text{Eve}}|\bar{\bm\phi})$ according to Equation \eqref{eq:KLD}, if there exists a unique vector $\bar{\bm\phi}^\prime\in\mathbb{R}^{2(K+1)\times1}$ such that $o(\bar{\bm\phi}^\prime) = u(\bm\phi^\star)$, we have $\bar{\bm\phi}^\star=\bar{\bm\phi}^\prime$ according to the non-negative property of KL divergence\footnote{An alternative proof of this statement can be found in \cite[Equations (13) and (14)]{ZhengmismatchRIS}, in terms of the KL divergence for two Gaussian distributions.}. Then, our goal amounts to deriving such a $\bar{\bm\phi}^\prime$. Since the path with the smallest TOA is typically treated as the LOS path for localization \cite{Ayyalasomayajula,FascistaMISOML}, considering the effect of the potential phase wrapping caused by the proposed DAIS scheme, we discuss the following two cases.

\begin{enumerate}
    \item[C1)] $\bar{\tau}_0=\min\{\bar{\tau}_0, \bar{\tau}_1,\cdots \bar{\tau}_K\}$ holds. In this case, the LOS path can be correctly distinguished if error in the estimated TOAs is small enough. By solving $o(\bar{\bm\phi}) = u(\bm\phi^{\star})$ for $\bar{\bm\phi}^{\star}$, we have the unique solution as \vspace{-8.5pt}

        \begin{equation}\label{eq:ptruecase1}
            \begin{aligned}
                \bar{\boldsymbol{p}}^{\star} &= \boldsymbol{z} - c{\bar\tau}_0[\cos({\bar\theta}_{\mathrm{Tx},0}),\sin({\bar\theta}_{\mathrm{Tx},0})]^{\mathrm{T}}\\
                \bar{\boldsymbol{v}}^\star_{k} &=\frac{1}{2}\bar{b}^{\star}_{k}[\cos({\bar{\theta}}_{\mathrm{Tx},k}),\sin({\bar{\theta}}_{\mathrm{Tx},k})]^{\mathrm{T}}+\bar{\boldsymbol{p}}^{\star}\vspace{-8pt}%\\
                %\bar{v}^\star_{k,y} &= \frac{1}{2}\bar{b}^{\star}_{k}\sin({\bar{\theta}}_{\mathrm{Tx},k})+\bar{p}^{\star}_{y},
            \end{aligned}
        \end{equation}
        \begin{equation}
           \begin{aligned}    
           &      \mbox{where} \; \;\bar{b}^{\star}_{k}\\
           &=\frac{\left(c{\bar\tau}_k\right)^2-\left(z_x-\bar{p}^{\star}_{x}\right)^2-\left(z_y-\bar{p}^{\star}_{y}\right)^2}{c{\bar\tau}_k-\left(z_x-\bar{p}^{\star}_{x}\right)\cos({\bar{\theta}}_{\mathrm{Tx},k})-\left(z_y-\bar{p}^{\star}_{y}\right)\sin({\bar{\theta}}_{\mathrm{Tx},k})}, \vspace{-3pt}
           \end{aligned}
        \end{equation}
        with $k=1,2,\cdots,K$.
 
    \item[C2)] $\bar{\tau}_0\neq\min\{\bar{\tau}_0, \bar{\tau}_1,\cdots \bar{\tau}_K\}$ holds. In this case, the LOS path cannot be correctly distinguished even with the perfect knowledge of the shifted TOAs. Without loss of generality, we assume that $\bar{\tau}_1=\min\{\bar{\tau}_0, \bar{\tau}_1,\cdots \bar{\tau}_K\}$. By solving $o(\bar{\bm\phi}) = u(\bm\phi^{\star})$ for $\bar{\bm\phi}^{\star}$, the pseudo-true locations of Alice and the first scatterer are given by \vspace{-4.5pt}
        \begin{equation}\label{eq:ptruecase2}
            \begin{aligned}
                \bar{\boldsymbol{p}}^{\star} &= \boldsymbol{z} - c{\bar\tau}_1[\cos({\bar\theta}_{\mathrm{Tx},1}),\sin({\bar\theta}_{\mathrm{Tx},1})]^{\mathrm{T}}\\
                \bar{\boldsymbol{v}}^\star_{1} &=\frac{1}{2}\bar{b}^{\star}_{0}[\cos({\bar{\theta}}_{\mathrm{Tx},0}),\sin({\bar{\theta}}_{\mathrm{Tx},0})]^{\mathrm{T}}+\bar{\boldsymbol{p}}^{\star}\vspace{-5pt}%\\
                %\bar{v}^\star_{1,y} &= \frac{1}{2}\bar{b}^{\star}_{0}\sin({\bar{\theta}}_{\mathrm{Tx},0})+\bar{p}^{\star}_{y}. 
            \end{aligned}
        \end{equation}
    The pseudo-true locations of the other scatterers are the same as those derived in Case C1. 
\end{enumerate}
Combining the above two cases leads to the desired statement.

\vspace{-7pt} 
\section{Proof of Proposition \ref{prop:degradation}}\label{sec:proofdegradation}
\vspace{-3pt} 
For the true and mismatched distributions $g_{\text{T}}(\hat{\bm\eta}_{\text{Eve}}|\bm\phi^\star)$ and $g_{\text{M}}(\hat{\bm\eta}_{\text{Eve}}|\bar{\bm\phi})$ given by Equations \eqref{eq:truemodel} and \eqref{eq:mismatchmodel}, it can be verified that $\bm\Psi_{\bar{\bm \phi}^\star}^{(\romannumeral 1)}$ and $\bm\Psi_{\bar{\bm \phi}^\star}^{(\romannumeral 2)}$ are positive semidefinite matrices according to Equations \eqref{eq:EveLB}, \eqref{eq:A} and \eqref{eq:B}, while $\operatorname{Tr}\left(\bm\Psi^{(\romannumeral 2)}_{\bar{\bm \phi}^\star}\right)$ is not relevant {to} $\sigma$ from Equations \eqref{eq:ptruecase1} and \eqref{eq:ptruecase2}. Hence, $\operatorname{Tr}\left(\bm\Psi^{(\romannumeral 1)}_{\bar{\bm \phi}^\star}\right)\geq0$ and $\operatorname{Tr}\left(\bm\Psi^{(\romannumeral 2)}_{\bar{\bm \phi}^\star}\right)\geq0$ hold and the goal amounts to proving that there is a constant $\sigma_0$ such that $
    \operatorname{Tr}\left(\bm\Psi^{(\romannumeral 2)}_{\bar{\bm \phi}^\star}\right)>\operatorname{Tr}(\bm\Xi_{\bm \phi^{\star}}) $
for any $0\leq\sigma\leq\sigma_0$. Then, from Equations \eqref{eq:FIMce}, \eqref{eq:efim}, and \eqref{eq:FIMloc}, we have $
\lim_{\sigma\downarrow0}\operatorname{Tr}(\bm\Xi_{\bm \phi^{\star}})=0,$
yielding the desired statement.

\vspace{-7pt} 
\section{Proof of Corollary \ref{coro:EveLBsimplifed}}\label{sec:proofEveLBsimplifed}
\vspace{-3pt} 
When the vector of the pseudo-true locations of Alice and scatterers $\bar{\bm\phi}^\star$ is set according to Equation \eqref{eq:ptrue}, it can be verified that $o(\bar{\bm\phi}^\star) = u(\bm\phi^{\star})$, which indicates that $g_{\text{T}}(\hat{\bm\eta}_{\text{Eve}}|\bm\phi^\star)=g_{\text{M}}(\hat{\bm\eta}_{\text{Eve}}|\bar{\bm\phi})$ for $\bar{\bm\phi}=\bar{\bm\phi}^\star$ according to Equations \eqref{eq:truemodel} and \eqref{eq:mismatchmodel}. Hence, $\bm B_{\bar{\bm \phi}^\star}[r,l] = - \bm A_{\bar{\bm \phi}^\star}[r,l]$ holds due to \vspace{-5pt}
\begin{equation}
\begin{aligned}
    &\bm A_{\bar{\bm \phi}^\star}[r,l]=\mathbb{E}_{g_{\text{T}}(\hat{\bm\eta}_{\text{Eve}}|\bm\phi^\star)}\left\{\frac{\partial^2}{\partial\bar{\bm\phi}^\star[r]\partial\bar{\bm\phi^\star}[l]}\log g_{\text{M}}(\hat{\bm\eta}_{\text{Eve}}|\bar{\bm\phi}^\star)\right\}\\
    &=\left(\frac{\partial}{\partial\bar{\bm\phi^\star}[l]}\left(\frac{\partial o(\bar{\bm\phi}^\star)}{\partial\bar{\bm\phi}^\star[r]}\right)\right)^{\mathrm{T}}\bm\Sigma_{\bar{\bm\eta}}^{-1}\mathbb{E}_{g_{\text{T}}(\hat{\bm\eta}_{\text{Eve}}|\bm\phi^\star)}\left\{\hat{\bm\eta}_{\text{Eve}}-o(\bar{\bm\phi}^\star)\right\}\\
    & - \left(\frac{\partial o(\bar{\bm\phi}^\star)}{\partial\bar{\bm\phi}^\star[r]}\right)^{\mathrm{T}}\bm\Sigma_{\bar{\bm\eta}}^{-1}\frac{\partial o(\bar{\bm\phi}^\star)}{\partial\bar{\bm\phi}^\star[l]}\overset{\mathrm{(a)}}{=}- \left(\frac{\partial o(\bar{\bm\phi}^\star)}{\partial\bar{\bm\phi}^\star[r]}\right)^{\mathrm{T}}\bm\Sigma_{\bar{\bm\eta}}^{-1}\frac{\partial o(\bar{\bm\phi}^\star)}{\partial\bar{\bm\phi}^\star[l]}%\vspace{-3pt}
\end{aligned}
\end{equation}
and \vspace{-3pt}
\begin{equation}
\begin{aligned}
    &\bm B_{\bar{\bm \phi}^\star}[r,l]\\
    &=\mathbb{E}_{g_{\text{T}}(\hat{\bm\eta}_{\text{Eve}}|\bm\phi^\star)}\left\{\frac{\partial\log g_{\text{M}}(\hat{\bm\eta}_{\text{Eve}}|\bar{\bm\phi}^\star) }{\partial\bar{\bm\phi}^\star[r]}\frac{\partial\log g_{\text{M}}(\hat{\bm\eta}_{\text{Eve}}|\bar{\bm\phi}^\star) }{\partial\bar{\bm\phi}^\star[l]}\right\},\\
    &= \left(\frac{\partial o(\bar{\bm\phi}^\star)}{\partial\bar{\bm\phi}^\star[r]}\right)^{\mathrm{T}}\bm\Sigma_{\bar{\bm\eta}}^{-1}\mathbb{E}_{g_{\text{T}}(\hat{\bm\eta}_{\text{Eve}}|\bm\phi^\star)}\left\{\left(\hat{\bm\eta}_{\text{Eve}}-o(\bar{\bm\phi}^\star)\right)\right.\\
    &\quad\quad\quad\quad\quad\quad\quad\quad\quad\quad\ \left.\times\left(\hat{\bm\eta}_{\text{Eve}}-o(\bar{\bm\phi}^\star)\right)^{\mathrm{T}}\right\}\bm\Sigma_{\bar{\bm\eta}}^{-1}\frac{\partial o(\bar{\bm\phi}^\star)}{\partial\bar{\bm\phi}^\star[l]}\\
    &\overset{\mathrm{(b)}}{=} \left(\frac{\partial o(\bar{\bm\phi}^\star)}{\partial\bar{\bm\phi}^\star[r]}\right)^{\mathrm{T}}\bm\Sigma_{\bar{\bm\eta}}^{-1}\frac{\partial o(\bar{\bm\phi}^\star)}{\partial\bar{\bm\phi}^\star[l]},\vspace{-3pt}
\end{aligned}
\end{equation}
where $\mathrm{(a)}$ and $\mathrm{(b)}$ follow from $o(\bar{\bm\phi}^\star) = u(\bm\phi^{\star}) = \bar{\bm \eta}$ given $\bm\epsilon\sim\mathcal{N}(\bm0,\bm\Sigma_{\bar{\bm\eta}})$. Then, the lower bound for the MSE of Eve's localization can be simplified into \vspace{-5pt}
\begin{equation}\label{eq:EveLBsimplifiedderivation}
\begin{aligned}
    &\bm\Psi_{\bar{\bm \phi}^\star}={{\bm A_{\bar{\bm \phi}^\star}^{-1}\bm B_{\bar{\bm \phi}^\star}\bm A_{\bar{\bm \phi}^\star}^{-1}}+{(\bar{\bm\phi}^\star-\bm\phi^{\star})(\bar{\bm\phi}^\star-\bm\phi^{\star})^\mathrm{T}}}\\
    &={{\left( \left(\frac{\partial o(\bar{\bm\phi}^\star)}{\partial\bar{\bm\phi}^\star}\right)^{\mathrm{T}}\bm\Sigma_{\bar{\bm\eta}}^{-1}\frac{\partial o(\bar{\bm\phi}^\star)}{\partial\bar{\bm\phi}^\star}\right)^{-1}}+{(\bar{\bm\phi}^\star-\bm\phi^{\star})(\bar{\bm\phi}^\star-\bm\phi^{\star})^\mathrm{T}}}\\
    &={{\left( \left(\frac{\partial \bar{\bm \eta}}{\partial\bar{\bm\phi}^\star}\right)^{\mathrm{T}}\bm\Sigma_{\bar{\bm\eta}}^{-1}\frac{\partial \bar{\bm \eta}}{\partial\bar{\bm\phi}^\star}\right)^{-1}}+{(\bar{\bm\phi}^\star-\bm\phi^{\star})(\bar{\bm\phi}^\star-\bm\phi^{\star})^\mathrm{T}}}
      \end{aligned}
    \end{equation}
    \begin{equation}
        \begin{aligned}
    &={{\left(\bm\Pi_{\bar{\bm \phi}^{\star}}^\mathrm{T}\bm J_{\bar{\bm \eta}}\bm\Pi_{\bar{\bm \phi}^{\star}}\right)^{-1}}+{(\bar{\bm\phi}^\star-\bm\phi^{\star})(\bar{\bm\phi}^\star-\bm\phi^{\star})^\mathrm{T}}}\\
    &={{\bm J_{\bar{\bm \phi}^{\star}}^{-1}}+{(\bar{\bm\phi}^\star-\bm\phi^{\star})(\bar{\bm\phi}^\star-\bm\phi^{\star})^\mathrm{T}}},\vspace{-3pt}
\end{aligned}
\end{equation}
where $\bm\Pi_{\bar{\bm \phi}^{\star}}\triangleq\frac{\partial \bar{\bm \eta}}{\partial\bar{\bm\phi}^\star}$, concluding the proof.

\vspace{-10pt}
\section{Proof of Proposition \ref{prop:lowerboundpos}}\label{sec:prooflowerboundpos}
\vspace{-3pt} 
With slight abuse of notation, we shall define the vector of the location-relevant channel parameters as $\bar{\bm\eta}\triangleq[\bar{\bm\eta}^{\mathrm{T}}_{k_{\min}
},\bar{\bm\eta}_1^{\mathrm{T}},\cdots,\bar{\bm\eta}_0^{\mathrm{T}},\cdots,\bar{\bm\eta}_K^{\mathrm{T}}]^{\mathrm{T}}\in\mathbb{R}^{2(K+1)\times1}$ with $\bar{\bm\eta}_k\triangleq[\bar{\tau}_k,\bar{\theta}_{\mathrm{Tx},k}]^{\mathrm{T}}\in\mathbb{R}^{2\times1}$ for $k=0,1,\cdots,K$. Given the orthogonality of the paths in Equation \eqref{eq:orthopath}, it can be verified that the effective FIM for the estimation of $\bar{\bm\eta}$ is a block diagonal matrix, %\vspace{-3pt}
\begin{equation}\label{eq:blockJ}
    \bm J_{\bar{\bm \eta}}= \begin{bmatrix}\bm J_{\bar{\bm \eta}_{k_{\min}}} & \bm 0 & \cdots & \bm 0\\
    \bm 0 & \bm J_{\bar{\bm \eta}_1} & \cdots & \bm 0\\
    \vdots & \vdots & \ddots & \vdots\\
    \bm 0 & \bm 0 & \cdots & \bm J_{\bar{\bm \eta}_K}\end{bmatrix}.\vspace{-3pt}
\end{equation}
In addition, we can express the inverse of the transformation matrix $\bm\Pi_{\bar{\bm \phi}^{\star}}$ as\vspace{-5pt}
\begin{equation}\label{eq:inversePi}
\begin{aligned}
    &\bm\Pi_{\bar{\bm \phi}^{\star}}^{-1}\overset{\mathrm{(c)}}{=}\left(\frac{\partial\bar{\bm\phi}^\star}{\partial \bar{\bm\eta}}\right)^{\mathrm{T}}=\begin{bmatrix}
    \bm T_{0,0} & \bm 0 & \cdots & \bm 0\\
    \bm T_{0,1}       & \bm T_{1,1} & \cdots & \bm 0\\
    \vdots      & \vdots      & \ddots & \vdots\\
    \bm T_{0,K}       & \bm 0       & \cdots & \bm T_{K,K}
    \end{bmatrix},\vspace{-5pt}
\end{aligned} 
\end{equation}
where $\mathrm{(c)}$ holds due to the multivariate inverse function theorem \cite{FascistaMISOML} and the matrix $\bm T_{k,k^\prime}$ is given by \vspace{-5pt}
\begin{equation}\label{eq:subT}
    \bm T_{k,k^\prime} \triangleq \left(\frac{\partial\bar{\bm v}_k^\star}{\partial \bar{\bm\eta}_{k^\prime}}\right)^{\mathrm{T}} = \begin{bmatrix} \frac{\partial \bar{v}_{x,k}^\star}{\partial \bar{\tau}_{k^\prime}} & \frac{\partial \bar{v}_{y,k}^\star}{\partial \bar{\tau}_{k^\prime}} \\ \frac{\partial \bar{v}_{x,k}^\star}{\partial \bar{\theta}_{\mathrm{Tx},k^\prime}} & \frac{\partial \bar{v}_{y,k}^\star}{\partial \bar{\theta}_{\mathrm{Tx},k^\prime}}\end{bmatrix},\vspace{-5pt}
\end{equation}
for $k,k^{\prime}=0,1,\cdots,K$, with $\bar{\bm v}^\star_0\triangleq\bar{\bm p}^\star$. Then, ${\bm J_{\bar{\bm \phi}^{\star}}^{-1}}$ can be derived as follows, \vspace{-5pt}
\begin{equation}
\begin{aligned}
    {\bm J_{\bar{\bm \phi}^{\star}}^{-1}}= {\left(\bm\Pi_{\bar{\bm \phi}^{\star}}^\mathrm{T}\bm J_{\bar{\bm \eta}}\bm\Pi_{\bar{\bm \phi}^{\star}}\right)^{-1}}= \bm\Pi_{\bar{\bm \phi}^{\star}}^{-1}\bm J^{-1}_{\bar{\bm \eta}}\left(\bm\Pi^{-1}_{\bar{\bm \phi}^{\star}}\right)^\mathrm{T}.\vspace{-5pt}
\end{aligned} 
\end{equation}
Following from Corollary \ref{coro:EveLBsimplifed} as well as Equations \eqref{eq:blockJ} and \eqref{eq:inversePi}, it can be verified that  \vspace{-5pt}
\begin{equation}\label{eq:EveLBtbsimplify}
\begin{aligned}
    \mathbb{E}\left\{\left\|\hat{\bm p}_\text{Eve}-{\bm p}^{\star}\right\|^2_2\right\} \geq \bm T_{0,0}\bm J^{-1}_{\bar{\bm \eta}_{k_{\min}}}\bm T^{\mathrm{T}}_{0,0}+\left\|\bar{\bm p}^\star-{\bm p}^{\star}\right\|^2_2.\vspace{-5pt}
\end{aligned}
\end{equation}
According to the definition for the effective FIM given in Equations \eqref{eq:efim} and that for the transformation matrix given in \eqref{eq:subT}, substituting Equation \eqref{eq:ptrue} into Equation \eqref{eq:EveLBtbsimplify} and leveraging the asymptotic property of the FIM shown in \cite[Lemma 1]{li2023fpi} yield the desired statement.

\vspace{-10pt}
\section{Proof of Proposition \ref{prop:strucleakge}}\label{sec:proofstrucleakge}
\vspace{-3pt} 
Considering the knowledge of the {structure of the precoder}, we denote by ${\varpi}^{(g,n)} \triangleq\boldsymbol{h}^{(n)}{{\bm \Phi}^{(n)}}{\boldsymbol{s}}^{(g,n)}$ the noise-free observation for the estimation of $\bm\chi$, with $g=1,2,\cdots,G$ and $n=0,1,\cdots,N-1$. It can be verified that\vspace{-3pt}
\begin{equation}
\begin{aligned}
\frac{\partial  \varpi^{(g,n)}}{\partial \Delta_\tau}
&= \sum_{k=0}^{K} \frac{\partial \varpi^{(g,n)}}{\partial {\tau}^\star_k} \\
&={-\frac{j 2\pi n}{N T_{s}}}\sum_{k=0}^{K}\gamma^{\star}_k e^{-\frac{j 2\pi n\bar{\tau}_k}{N T_{s}}}\boldsymbol{ \alpha}\left(\bar{\theta}_{\mathrm{Tx},k}\right)^{\mathrm{H}}{\bm s}^{(g,n)},\vspace{-3pt}
\end{aligned}
\end{equation}
and \vspace{-3pt}
\begin{equation}
\begin{aligned}
\frac{\partial  \varpi^{(g,n)}}{\partial \Delta_{\theta}}&= \sum_{k=0}^{K} \frac{\cos(\Delta_{{\theta}_{\text{Tx}}})}{\cos({\theta}^\star_{\mathrm{Tx},k})}\frac{\partial \varpi^{(g,n)}}{\partial {\theta}^\star_{\mathrm{Tx},k}}\\
& = \frac{j2\pi d}{\lambda_c}\sum_{k=0}^{K}{\gamma}_k e^{-j\frac{2\pi n \bar{\tau}_k}{NT_s}}\cos(\Delta_{{\theta}_{\text{Tx}}}){\bm\alpha({\theta}^\star_{\text{Tx},k})^{\mathrm{H}}}\\
&\ \ \times \operatorname{diag}([0,1,\cdots,N_t-1])\operatorname{diag}\left(\bm \alpha\left({\Delta}_{\theta_\text{Tx}}\right)^{\mathrm{H}} \right)\bm s^{(g,n)}\vspace{-3pt}
\end{aligned}
\end{equation}
hold so there are two rows of $\bm J_{{\bm \chi}}$ {that} are linearly dependent on the others, which concludes the proof. 

\vspace{-10pt} 
\renewcommand*{\bibfont}{\footnotesize}
%\addtocategory{bluepapers}{Schaefer,ZhangQ}
\printbibliography
\end{document}